\documentclass[11pt,letterpaper]{article}

\usepackage{comment}
\usepackage{amsmath}
\usepackage{amsfonts}
\usepackage{amsthm}
\usepackage{amssymb}
\usepackage{thm-restate}
\usepackage{thmtools}
\usepackage{mathtools}
\usepackage{pgffor}
\usepackage{xcolor}
\usepackage[lined,boxed,ruled,norelsize,algo2e,linesnumbered,noend]{algorithm2e}
\usepackage[]{hyperref}
\usepackage{cleveref}
\usepackage{geometry}
\usepackage{graphicx}
\usepackage{enumitem}
\usepackage{url} 

\usepackage[backend=biber, isbn=false, style=alphabetic, backref=true,
doi=false, url=false, maxcitenames=10, mincitenames=5, maxbibnames=10, minbibnames=5, minalphanames=3,
defernumbers=true, sorting=ynt, sortcites]{biblatex}
\hypersetup{colorlinks,linkcolor={blue},citecolor={blue},urlcolor={red}}

\crefname{subsection}{subsection}{subsections}
\crefname{section}{section}{sections}

\makeatletter
\newcommand\footnoteref[1]{\protected@xdef\@thefnmark{\labelcref{#1}}\@footnotemark}
\makeatother

\newtheorem{theorem}{Theorem}[section]
\newtheorem{lemma}[theorem]{Lemma}

\newtheorem{definition}[theorem]{Definition}

\newcommand{\R}{\mathbb{R}}

\newcommand{\Z}{\mathbb{Z}}

\newcommand{\E}{\mathbb{E}}

\renewcommand{\tilde}{\widetilde}
\renewcommand{\hat}{\widehat}
\renewcommand{\bar}{\overline}

\newcommand{\poly}{\operatorname{poly}}

\newcommand{\nnz}{\operatorname{nnz}}%
\newcommand{\cnorm}{C_{\mathrm{norm}}}%

\newcommand{\tO}{{\tilde O}}
\newcommand{\Otil}{{\tilde O}}

\foreach \x in {A,...,Z}{%
	\expandafter\xdef\csname m\x\endcsname{\noexpand\mathbf{\x}}
}
\newcommand{\ma}{\mathbf{A}}%

\newcommand{\otau}{\noexpand{\overline{\tau}}}
\newcommand{\omu}{\noexpand{\overline{\mu}}}
\newcommand{\osigma}{\noexpand{\overline{\sigma}}}%
\newcommand{\os}{\noexpand{\overline{s}}}
\newcommand{\ot}{\noexpand{\overline{t}}}
\newcommand{\og}{\noexpand{\overline{g}}}
\newcommand{\ou}{\noexpand{\overline{u}}}
\newcommand{\ov}{\noexpand{\overline{v}}}
\newcommand{\ow}{\noexpand{\overline{w}}}
\newcommand{\ox}{\noexpand{\overline{x}}}
\newcommand{\oy}{\noexpand{\overline{y}}}
\newcommand{\oz}{\noexpand{\overline{z}}}

\foreach \x in {A,...,Z}{%
	\expandafter\xdef\csname c\x\endcsname{\noexpand\mathcal{\x}}
}

\foreach \x in {a,...,z}{%
	\expandafter\xdef\csname t\x\endcsname{\noexpand\widetilde{\x}}
}

\foreach \x in {A,...,Z}{%
	\expandafter\xdef\csname om\x\endcsname{\noexpand\mathbf{\overline{\x}}}
}

\foreach \x in {A,...,Z}{%
	\expandafter\xdef\csname tm\x\endcsname{\noexpand\mathbf{\widetilde{\x}}}
}

\newcommand{\tmp}{{\mathrm{(tmp)}}} %
\newcommand{\new}{\mathrm{(new)}}%
\newcommand{\init}{\mathrm{(init)}}%
\newcommand{\target}{\mathrm{(end)}}%
\newcommand{\final}{\mathrm{(final)}}%
\newcommand{\apxfinal}{\mathrm{(apx\text{-}final)}}%

\newcommand{\opt}{\mathrm{OPT}}

\newcommand{\norm}[1]{\|#1\|}

\newcommand{\defeq}{\stackrel{\mathrm{{\scriptscriptstyle def}}}{=}}

\newcommand{\Comment}[1]{\tcp*[h]{#1}}

\newcommand{\Tau}{\mathbf{T}}

\newcommand{\eps}{\varepsilon}
\renewcommand{\d}{\delta}
\newcommand{\diag}{\mathbf{diag}}
\newcommand{\g}{\nabla}

\newcommand{\ShortStep}{\textsc{ShortStep}}
\newcommand{\TwoPartyShortStep}{\textsc{TwoPartyShortStep}}
\newcommand{\PathFollowing}{\textsc{PathFollowing}}
\newcommand{\SpecApprox}{\textsc{SpectralApprox}}

\newcommand{\hw}{\hat{w}}

\newcommand{\ls}{\lesssim}

\newcommand{\wt}{\widetilde}

\newcommand{\cvalid}{C_{\mathrm{valid}}}
\newcommand{\Var}{\mathrm{Var}}
\newcommand{\assign}{\leftarrow}
\newcommand{\cstart}{C_{\mathrm{start}}}

\newcommand{\twopartdef}[4]
{
\left\{
\begin{array}{ll}
#1 & \mbox{} #2 \\
#3 & \mbox{} #4
\end{array}
\right.
}

\newcommand{\pA}{^{(Alice)}}
\newcommand{\pB}{^{(Bob)}}

\title{A Subquadratic Two-Party Communication Protocol for Minimum Cost Flow
}

\author{
Hossein Gholizadeh\thanks{Karlsruhe Institute of Technology and Heidelberg University, {hgholizadeh8@gmail.com}}\and
Yonggang Jiang\thanks{MPI-INF and Saarland University, Germany, {yjiang@mpi-inf.mpg.de}}
\and
}
\date{}

\begin{document}

\maketitle

\begin{abstract}

In this paper, we discuss the \textbf{maximum flow problem} in the \textbf{two-party communication model}, where two parties, each holding a subset of edges on a common vertex set, aim to compute the maximum flow of the union graph with minimal communication. We show that this can be solved with $\tilde{O}(n^{1.5})$ bits of communication, improving upon the trivial $\tilde{O}(n^2)$ bound.

To achieve this, we derive two additional, more general results:
\begin{enumerate}
    \item We present a randomized algorithm for \textbf{linear programs} with two-sided constraints that requires $\tilde{O}(n^{1.5}k)$ bits of communication when each constraint has at most $k$ non-zeros. This result improves upon the prior work by \cite{ghadiri2024improvingbitcomplexitycommunication}, which achieves a complexity of $\tilde{O}(n^2)$ bits for LPs with one-sided constraints. Upon more precise analysis, their algorithm can reach a bit complexity of $\tilde{O}(n^{1.5} + nk)$ for one-sided constraint LPs. Nevertheless, for sparse matrices, our approach matches this complexity while extending the scope to two-sided constraints.
    
    \item Leveraging this result, we demonstrate that the \textbf{minimum cost flow problem}, as a special case of solving linear programs with two-sided constraints and as a general case of maximum flow problem, can also be solved with a communication complexity of $\tilde{O}(n^{1.5})$ bits.
\end{enumerate}

These results are achieved by adapting an interior-point method (IPM)-based algorithm for solving LPs with two-sided constraints in the sequential setting by \cite{brand2021minimumcostflowsmdps} to the two-party communication model. This adaptation utilizes techniques developed by \cite{ghadiri2024improvingbitcomplexitycommunication} for distributed convex optimization.

\end{abstract}

\clearpage
\tableofcontents

\clearpage

\section{Introduction} \label{section:Introduction}

In the maximum flow (maxflow) problem, we are given a connected directed graph \( G = (V, E, u) \) with \( n \) vertices and \( m \) edges, where the edges have non-negative capacities \( u \in \R_{\geq 0}^m \). The objective is to maximize the amount of flow routed between two designated vertices, the source \( s \) and the sink \( t \), while ensuring that the flow through each edge does not exceed its capacity.

A more general variant of the maximum flow problem is the minimum cost flow (mincost flow) problem. In this problem, in addition to edge capacities, edge costs \( c \in \R^m \) are introduced, and instead of focusing on a specific \( s \)-\( t \) pair, vertex demands \( d \in \R^n \) are defined. These demands generalize \( s \)-\( t \) flows by allowing any vertex to act as a source or sink based on its demand value. The goal is to find a feasible flow that satisfies the vertex demands and capacity constraints on the edges while minimizing the total cost of the flow.

These extensively studied problems are fundamental to various combinatorial and numerical tasks and are central problems in computer science and economics. They serve as key solutions for problems like finding maximum bipartite matching, determining the minimum \( s \)-\( t \) cut in a network, computing shortest paths in graphs with negative edge weights, and solving the transshipment problem (see e.g. \cite{brand2021bipartitematchingnearlylineartime, brand2021minimumcostflowsmdps}). 

Traditional algorithms for solving these problems rely on iterative improvements to flows through fundamental primitives like augmenting paths and blocking flows (see, for example, \cite{Karzanov1973OnFA, TarjanNetworkFlow, Goldberg1990FindingMC, Edmonds2003}). However, the past decade has seen substantial progress in runtime efficiency with the introduction of algorithms based on interior-point methods (IPM) (e.g., \cite{lee2015pathfindingisolving, kathuria2020potentialreductioninspiredalgorithm, brand2021minimumcostflowsmdps, chen2022maximumflowminimumcostflow}). These modern approaches primarily address the optimization of linear programs of the form:  
\begin{align}
    \min_{\substack{x \in \R^m : \mA^\top x = b \\ \forall i \in \{1,\dots,m \} : \ell_i \le x_i \le u_i}} c^\top x, \label{eq:generalLPIntro}
\end{align}

where \( b \in \R^n \), \( c \in \R^m \), \( \mA \in \R^{m \times n} \), and \( \ell_i\), \( u_i \in \R \) with \(\ell_i \leq u_i\) for all \( i \in \{1, \dots, m\} \). Notably, by setting \( \mA \) as the incidence matrix of a graph, \( b \) as the vertex demands, \( \ell_i = 0 \) for all \( i \), \( u_i \) as the edge capacities, and \( c \) as the edge cost vector, this linear program (LP) directly corresponds to the mincost flow problem.

The fastest known algorithm for the mincost flow problem in the sequential setting is a randomized method by \cite{chen2022maximumflowminimumcostflow}, which computes an exact mincost flow in \( O(m^{1 + o(1)} \log W) \) time. Given that the problem inherently requires \( \Omega(m) \) time due to the number of edges, this runtime approaches what is likely to be optimal. However, it remains unclear how these advancements can be adapted to other computational models, such as two-party communication, streaming, query, quantum, or parallel frameworks. This thesis focuses on addressing this gap by exploring how these results can be extended to the two-party communication model.

To be specific, our goal is to determine the classical communication complexity of the minimum cost flow problem. In this setting, we assume the edges of the input graph \( G \) are distributed between two parties, Alice and Bob. The objective is for Alice and Bob to collaboratively compute a minimum cost flow in their combined graph while minimizing the communication between them. Here, we assume both parties have unlimited computational power locally, so the primary cost is the number of bits exchanged during communication.

The two-party communication model is a fundamental framework introduced by \cite{Yao1979SomeCQ} in the late 1970s. It was initially motivated by applications in VLSI design, where communication complexity provides direct lower bounds for measures like the minimum bisection width of a chip, as well as for its area-delay squared product \cite{Thomborson1980ACT}. These connections between communication complexity and hardware design underscored the practical importance of understanding communication costs. Over time, the model evolved into a central tool in computational complexity, providing a structured framework for classifying problems and studying trade-offs between communication and computation.

In the context of graph algorithms, the two-party communication model has been extensively studied over the past four decades. Many fundamental graph problems have been explored in this framework, such as bipartite matching, connectivity, and planarity testing (e.g., \cite{Babai1986ComplexityCI, Papadimitriou1984CommunicationC, Hajnal1988OnTC, CommunicationComplexityPlanarity, dobzinski2014economicefficiencyrequiresinteraction, assadi_et_al:LIPIcs.APPROX/RANDOM.2021.48}). However, efficiently solving the maximum flow problem (or equivalently, finding an $s$-$t$ minimum cut) remains an unresolved challenge in this model.

Recent progress has been made on special cases of this problem. For example, \cite{blikstad2022nearlyoptimalcommunicationquery} demonstrated that the maximum-cardinality bipartite matching problem (BMM) can be solved in this model using just \( O(n \log^2 n) \) bits of communication, significantly improving on the prior upper bound of \( \tO(n^{1.5}) \) bits. \cite{Hajnal1988OnTC, huang2017communicationcomplexityapproximatemaximum} also established lower bounds of \( \Omega(n) \) and \( \Omega(n \log n) \) for randomized and deterministic solutions to BMM, respectively. Since BMM can be reduced to maxflow, these lower bounds also apply to the maxflow problem.

While \cite{blikstad2022nearlyoptimalcommunicationquery}'s techniques for BMM cannot be directly extended to maxflow, their work raises the question of whether we can achieve an upper bound for maxflow that improves upon the trivial bound of \( \tO(n^2) \), which involves Alice and/or Bob simply sending all their edges to the other party.

\subsection{Our Results} \label{subsection:ourResults}

In this subsection, we present the main results of this thesis. Note that our primary contribution is an efficient algorithm for solving a linear program (LP) with two-sided constraints in the two-party communication model, i.e., LPs of the following form: 

\begin{align}
    \min_{\substack{\mA^\top x = b \\ \ell \le x \le u}} c^\top x, 
\end{align}

where $\mA \in \R^{m \times n}$, $c \in \R^m$, $b \in \R^n$, and $\ell, u \in \R^m$ define the lower and upper bounds for $x \in \R^m$, respectively. Specifically, Alice and Bob each hold different rows of $\mA$, denoted as $\mA\pA$ and $\mA\pB$, along with corresponding portions of $c$, $\ell$, and $u$. Both parties know $b$. Furthermore, we assume that the bit complexity of each entry of $\mA$, $b$, and $c$ is bounded by $L$. Then, for any constant $\delta > 0$, known by Alice and Bob, we have:

\begin{restatable}[Linear Programming in the Two-Party Communication Model]{theorem}{theoremgeneralLP} \label{theorem:generalLP} 
    For any constant $\delta > 0$, there exists a randomized algorithm in the two-party communication setting that, given $\mA \in \R^{m \times n}$, $c, \ell, u \in \R^m$, and $b \in \R^n$, and assuming there exists a point $x \in \R^m$ such that $\mA^\top x = b$ and $\ell \leq x \leq u$, and furthermore assuming that the entries of $\mA$, $c$, $\ell$, and $u$ are bounded by $L$, computes, with high probability, a vector $x^\final \in \R^m$ such that:
    \[
    \norm{\mA^\top x^\final - b}_\infty \leq \delta, \quad \ell \leq x^\final \leq u, \quad \text{and} \quad c^\top x^\final \leq \min_{\substack{\mA^\top x = b \\ \ell \le x \le u}} c^\top x + \delta.
    \]
    The algorithm requires 
    \[
    \tO(n^{1.5} L^2 (k + \log \kappa \log m) \log m) \text{ bits of communication},
    \]
    where $k = \max_{i \in [m]}(\nnz(a_i))$ is the upper bound on the number of non-zero entries in each row of $\mA$, and $\kappa$ is the condition number of $\mA$, which intuitively measures how sensitive the solution of a system involving $\mA$ is to small changes in its input or coefficients.
\end{restatable}

Note that there are similar results concerning solving general LPs in this setting. For example, \cite{ghadiri2024improvingbitcomplexitycommunication} show that LPs with one-sided constraints, i.e., $x \geq 0$ instead of $\ell \leq x \leq u$, can be solved with $\tO(s n^{1.5}L + n^2L)$ bits of communication in the coordinator model. The coordinator model is essentially an extension of the two-party communication model, with multiple (here we assume $s$) communication parties instead of just two. To achieve their result, they demonstrate that by applying certain modifications to the interior-point method (IPM) introduced by \cite{tallDense}, the algorithm can be adapted to work in the coordinator model. 

However, their approach does not address LPs with two-sided constraints. In this work, we take a different approach. By leveraging the IPM developed by \cite{brand2021minimumcostflowsmdps}, we demonstrate that the techniques of \cite{ghadiri2024improvingbitcomplexitycommunication} can be applied in a novel way to also solve LPs with two-sided constraints effectively.

Furthermore, we demonstrate that our interior-point method (IPM) can be utilized to solve instances of the minimum cost flow problem. In the two-party communication model, we assume that each party knows a subset of the edges in the graph, along with the respective capacities and edge costs. Both parties know the vertex set of the graph and their respective demands. The objective is for the two parties to collaboratively find a feasible minimum cost flow in their union graph while minimizing communication. In this context, we derive the following result:

\begin{restatable}[Minimum Cost Flow in the Two-Party Communication Model]{theorem}{theoremmincostflow}
\label{theorem:mincostflow}
There exists a randomized algorithm in the two-party communication model that, with high probability, computes a minimum cost flow \( f \in \Z^m \) on a \( n \)-vertex, \( m \)-edge directed graph \( G = (V, E, u, c, d) \), where:
\begin{itemize}
    \item \( u \in \Z^m_{\geq 0} \) represents the integral edge capacities,
    \item \( c \in \Z^m \) denotes the integral edge costs, and
    \item \( d \in \Z^n \) specifies the integral vertex demands.
\end{itemize}
The algorithm communicates at most:
\[
\tO( n^{1.5} \log^2 ( \|u\|_\infty \|c\|_\infty ) ) \text{ bits}.
\]
\end{restatable}

As mentioned earlier, maxflow problem can be viewed as a special case of the micost flow problem. Thus, in the two-party communication model, it is assumed that, similar to the mincost flow problem, each party knows a subset of edges along with their respective capacities. Consequently, by using the algorithm designed for the mincost flow problem, it is possible to solve the maxflow problem with a communication complexity of $\tO(n^{1.5} \log^2 \|u\|_\infty)$ bits (we discuss how this is possible in \Cref{subsection:maxflow}). 

However, using standard capacity scaling methods, it is possible to shave off a logarithmic factor and achieve a yet better communication complexity, namely:

\begin{restatable}[Maximum Flow in the Two-Party Communication Model]{theorem}{theoremmaxflow} \label{theorem:maxflow}
There exists a randomized algorithm in the two-party communication model that, with high probability, computes a maximum flow \( f \in \Z^m \) on a \( n \)-vertex, \( m \)-edge directed graph \( G = (V, E, u) \), where \( u \in \Z^m_{\geq 0} \) are the integral edge capacities. The algorithm communicates at most:
\[
\tO( n^{1.5} \log \|u\|_\infty ) \text{ bits}.
\]
\end{restatable}

For this result, we orient ourselves on the procedure explained by \cite{brand2021minimumcostflowsmdps}. We explain this procedure formally in \Cref{subsection:maxflow}.

\paragraph{A concurrent work \cite{BrandWZ25}.} Independently, the streaming algorithm in \cite{BrandWZ25} implies the same communication complexity for minimum-cost flow and maximum flow in the two-party communication model. Nonetheless, our result can be seen as a more straightforward way to obtain the communication complexity.

\subsection{Organization} \label{subsection:orga}

Here, we present a short review of the structure of this work.

First, in \Cref{section:prelims}, we introduce the notation and preliminaries that will be used throughout the text. 

In \Cref{section:techoverview}, we provide a technical overview of our approach. We begin with a discussion of linear programming in the sequential setting, reviewing the interior point method (IPM) of \cite{brand2021minimumcostflowsmdps} (\Cref{subsection:lpinseq}). We then identify the main bottlenecks that arise when adapting this algorithm to the two-party communication model and explain the tools we use to overcome them (\Cref{subsection:bottlenecks}). Finally, we outline the precise goals of this work (\Cref{subsection:maingoals}). 

In \Cref{section:LPinCC}, we present and prove our main result on the communication complexity of solving linear programs with two-sided constraints. We first describe how to compute $\ell_p$-Lewis weights and spectral approximations in the two-party communication setting, which are necessary for adapting the sequential IPM (\Cref{subsec:lewisWeightsAndSpecApprox}). We then provide a detailed account of the path-following algorithm (\Cref{subsection:pathfollowing}), explain how to construct an initial feasible point by reducing to a modified LP (\Cref{subsection:initialpoint}), and finally prove the correctness and communication complexity of the resulting protocol (\Cref{subsection:proofIPM}), thereby establishing our main claim.

In \Cref{section:MCMFinTPCModel}, we turn to the minimum cost flow problem. We begin by presenting the analysis tools required for this setting. We then prove the correctness and communication complexity of our algorithm for minimum cost flow (\Cref{subsection:mincostproof}), and finally demonstrate how the same techniques yield a protocol for the maximum flow problem as a special case (\Cref{subsection:maxflow}). 

 \clearpage
\section{Preliminaries} \label{section:prelims}

We primarily follow the notation used in \cite{brand2021minimumcostflowsmdps} and \cite{ghadiri2024improvingbitcomplexitycommunication}. Let $[n] \defeq {1, 2, \ldots, n}$ denote the set of the first $n$ natural numbers. The notation $\tO(\cdot)$ is used to hide polylogarithmic factors, i.e., factors of the form $(\log n)^{O(1)}$, as well as $(\log \log W)^{O(1)} $ and $\log \varepsilon^{-1}$ factors. The term \emph{with high probability} (abbreviated as \emph{w.h.p.}) refers to a probability of at least $1 - n^{-c}$ for some constant $c > 0$.

\paragraph{Matrices.} For any matrix $\mA \in \R^{m \times n}$ or vector $v \in \R^d$, let $\mA\pA \in \R^{\Tilde{m} \times n}$ and $v\pA \in \R^{\Tilde{d}}$ denote the portions of $\mA$ and $v$ that are stored by Alice, respectively. Similarly, we define $\mA\pB$ and $v\pB$ as the corresponding parts stored by Bob. We also denote $a_i \in \R^n$ as the $i^\mathrm{th}$ row of $\mA$, represented as a column vector.

Additionally, given a vector $v \in \R^d$, we define $\mV \in \R^{d \times d}$ as the diagonal matrix whose diagonal entries are the elements of $v$, i.e., $\mV_{i,i} = v_i$ for all $i \in [d]$.

\paragraph{Matrix and Vector Operations.}
Given vectors $u, v \in \R^d$ for some $d$, the arithmetic operations $\cdot, +, -, /, \sqrt{\cdot}$ are performed element-wise. For example, $(u \cdot v)_i = u_i \cdot v_i$ and $(\sqrt{v})_i = \sqrt{v_i}$. Similarly, for a scalar $\alpha \in \R$, we define $(\alpha v)_i = \alpha v_i$ and $(v + \alpha)_i = v_i + \alpha$.

For a positive definite matrix $\mM \in \R^{n \times n}$, we define the weighted Euclidean $\mM$-norm of a vector $x$ as $\norm{x}_\mM = \sqrt{x^\top \mM x}$. Furthermore, for symmetric matrices $\mA, \mB \in \R^{n \times n}$, we use $\preceq$ to denote the Loewner ordering, i.e., $\mB \preceq \mA$ if and only if $\norm{x}_{\mA - \mB} \geq 0$ for all $x \in \R^n$.

In this context, we write $\mA \approx_\varepsilon \mB$ if and only if $\exp(-\varepsilon) \mB \preceq \mA \preceq \exp(\varepsilon) \mB$. Similarly, we extend this notation for vectors, letting $u \approx_\varepsilon v$ if and only if $\exp(-\varepsilon) v \leq u \leq \exp(\varepsilon) v$ entrywise. This implies that $u \approx_\varepsilon v \approx_\delta w$ yields $u \approx_{\varepsilon + \delta} w$, and $u \approx_\varepsilon v$ implies $u^\alpha \approx_{\varepsilon \cdot |\alpha|} v^\alpha$ for any $\alpha \in \R$.

The condition number of a matrix $\mA$ is defined as $\kappa(\mA) = \norm{\mA}_2 \cdot \norm{\mA^{-1}}_2$, where $\norm{\mA}_2$ denotes the operator norm of $\mA$, i.e., its largest singular value.

\paragraph{Highly 1-Self-Concordant Barrier Functions \(\phi_i\).}  
As demonstrated by \cite{brand2021minimumcostflowsmdps}, the path-following IPM requires highly 1-self-concordant barrier functions \(\phi_i: (\ell_i, u_i) \rightarrow \R\) for $i \in [m]$.  

For an interval \((\ell, u)\), a function \(f: (\ell, u) \rightarrow \R\) is defined to be a highly 1-self-concordant barrier on \((\ell, u)\) if, for all \(x \in (\ell, u)\), the following conditions are satisfied:  
\[
|f'(x)| \leq f''(x)^{1/2}, \quad |f'''(x)| \leq 2f''(x)^{3/2}, \quad |f''''(x)| \leq 6 f''(x)^2, \quad \text{and} \quad \lim_{x \rightarrow b} f(x) = +\infty.  
\]  

\cite{brand2021minimumcostflowsmdps} set these functions to \(\phi_i(x) = -\log(x - \ell_i) - \log(u_i - x)\) for $i \in [m]$ and prove that these functions are highly 1-self-concordant (see Lemma 4.3 of \cite{brand2021minimumcostflowsmdps}).  

Additionally, the first and second derivatives of \(\phi_i\) are given by:  
\[
\phi'_i(x) = -\frac{1}{x - \ell_i} + \frac{1}{u_i - x}, \quad \text{and} \quad \phi''_i(x) = \frac{1}{(x - \ell_i)^2} + \frac{1}{(u_i - x)^2}.
\]  

Furthermore, for \(x \in \R^m\), \(\phi(x) \in \R^m\) is the vector obtained by applying \(\phi_i\) to \(x_i\) for each \(i \in [m]\). Similarly, \(\phi'(x)\) and \(\phi''(x)\) are defined by applying \(\phi'_i\) and \(\phi''_i\) element-wise to \(x\), respectively. As described in the previous paragraph, \(\Phi\), \(\Phi'\), and \(\Phi''\) are diagonal matrices constructed using these vectors.

\paragraph{Leverage Scores and Spectral Approximation.}
We say that a matrix $\tilde{\mA} \in \R^{\tilde{m} \times n}$ is a spectral approximation of a matrix $\mA \in \R^{m \times n}$ if and only if $\tilde{\mA}^\top \tilde{\mA} \approx_\varepsilon \mA^\top \mA$ for some $\varepsilon > 0$.

Note that there are alternative notations to define spectral approximation. For example, \cite{ghadiri2024improvingbitcomplexitycommunication} defines $\tilde{\mA}$ as a $\lambda$-spectral approximation of $\mA$ if $\frac{1}{\lambda} \mA^\top \mA \preceq \tilde{\mA}^\top \tilde{\mA} \preceq \mA^\top \mA$. Since these notations are equivalent, we may interchange between them throughout the proof for convenience. Unless stated otherwise, we use the notation from the previous paragraph.

A fundamental tool for obtaining a spectral approximation with a small number of rows is the concept of leverage scores.

For a full-rank matrix $\mA \in \R^{m \times n}$, let $\sigma(\mA) \in \R^m$ denote its leverage scores, defined as $\sigma(\mA)_i \defeq a_i^\top (\mA^\top \mA)^{-1} a_i$ for each $i \in [m]$. It is known (see e.g. \cite{spielman2009graphsparsificationeffectiveresistances}) that sampling $O(n \log(n))$ rows with probability proportional to their leverage scores yields a spectral approximation of the original matrix.

\paragraph{$\ell_p$-Lewis Weights.}
For $p \in (0, \infty)$ and a full-rank matrix $\mA \in \R^{m \times n}$, the $\ell_p$-Lewis weights are defined as the solution $w \in \R^m_{>0}$ to the equation $w = \sigma(\mW^{\frac{1}{2} - \frac{1}{p}} \mA)$, where $\mW = \text{diag}(w)$.

\paragraph{Bit Complexity.}
In fixed-point arithmetic, a number is represented using $L$ bits if it has at most $L$ bits before the decimal point and at most $L$ bits after the decimal point. Such a number is thus in the set $\{0\} \cup [2^{-L}, 2^{-L} - 1] \cup [-2^{L} + 1, - 2^{-L}]$.

Furthermore, as mentioned before, our main focus in this work is the communication complexity of linear programming, the mincost flow problem, and the maxflow problem in the two-party communication model. Below, we define this model more specifically:

\paragraph{Two-Party Communication Model.}  
In this setup, there are two communication parties, conventionally referred to as Alice and Bob. Together, they aim to solve a problem under the assumption that each holds partial information about the problem. Both Alice and Bob are assumed to have infinite computational power, meaning the internal computation time is not considered in the analysis. Instead, the cost we seek to minimize is the amount of communication between them. The goal is for Alice and Bob to solve the problem with as little communication as possible. 

In the subsequent sections, we specify, for each problem, the type of information held by Alice and Bob and how their data is partitioned.

 \clearpage

\section{Technical Overview} \label{section:techoverview}

The main idea of this work is to adapt the algorithm from \cite{brand2021minimumcostflowsmdps} using some tools from \cite{ghadiri2024improvingbitcomplexitycommunication}, in order to design a protocol for solving max-flow in the two-party communication setting.

In the overview, we first give a short introduction to the algorithm from \cite{brand2021minimumcostflowsmdps}, then explain the main bottlenecks that appear when we attempt to convert this algorithm from the sequential setting to the communication setting, and finally introduce some tools from \cite{ghadiri2024improvingbitcomplexitycommunication}, which we later use to handle these bottlenecks.

If you are already familiar with these works and their results, you can skip the corresponding parts, as they do not contain new information.

\subsection{Linear Programming in the Sequential Setting} 
\label{subsection:lpinseq}

As is well known, the max-flow problem can be generalized to the min-cost flow problem, which in turn can be generalized to linear programs. \cite{brand2021minimumcostflowsmdps} introduce an algorithm for solving linear programs of the following form in the sequential setting:
\[
\min_{\substack{\mA^\top x = b \\ \ell \le x \le u}} c^\top x.
\]

The main contribution of \cite{brand2021minimumcostflowsmdps} is a data structure, called the \emph{HeavyHitter} data structure. This is crucial in the sequential setting since it reduces the internal running time, but it is not relevant in the communication model. Nevertheless, their algorithm still serves as a good template for an interior-point method for solving linear programs. In the following, we give an overview of how this IPM works.

\paragraph{Intuition Behind the Approach.}  
The IPM primarily consists of two main components:  
\begin{enumerate}  
    \item \emph{Finding an Initial Feasible Point}: This involves modifying the original LP into a related problem with a trivial feasible solution. The details of this process are discussed in \Cref{subsection:initialpoint}.  
    \item \emph{Path-Following Algorithm}: Starting from the initial feasible point, the algorithm iteratively follows a path towards a near-optimal, near-feasible solution. Over the course of \(\tO(\sqrt{n})\) iterations, in order to make a step to the next point, the algorithm solves Laplacian systems of the form \((\mA^\top \mD \mA)x = b\) in each iteration, where \(\mD\) is a positive definite diagonal matrix.  
\end{enumerate}  

The main computational challenge is analyzing the path-following algorithm to adapt it for the two-party communication model. Below, we provide a simplified overview of this algorithm:

\paragraph{Maintaining the Triple \((x, s, \mu)\).}  
In each iteration, the path-following algorithm maintains the triple \((x, s, \mu)\), where:  
\begin{itemize}  
    \item \textbf{\(\mu\)} is the path parameter, which gradually decreases to improve the solution's proximity to optimality.  
    \item \textbf{\(x\)} represents the primal variable, which is updated to improve optimality while maintaining feasibility.  
    \item \textbf{\(s\)} represents the dual slack variable, which helps in keeping track of the primal-dual gap.  
\end{itemize}  

Formally, as mentioned by \cite{brand2021minimumcostflowsmdps}, to satisfy feasibility conditions and optimality with regard to $\mu$, the constraints 
\[
\mA^\top x = b, \quad s + \mu \mW(x) \nabla \Phi(x) = 0, \quad \text{and} \quad s = c + \mA y,
\]
for some \(y \in \R^n\), must hold at each iteration. Here, \(w\) is a weight function, and its choice is crucial for guiding the triple along the central path and improving convergence rates in each iteration.

The approach proposed by \cite{brand2021minimumcostflowsmdps} allows \(w\) to depend on \(x\). For this purpose, they utilize \(\ell_p\)-Lewis weights, defined as:
\[
w(x) = \sigma\left(\mW(x)^{\frac{1}{2} - \frac{1}{p}} \left(\nabla^2 \Phi(x)\right)^{-\frac{1}{2}} \mA\right),
\]
where \(\sigma(\cdot)\) represents leverage scores. This choice enables the solution of linear programs in \(\tilde{O}(\sqrt{n})\) steps.

\paragraph{Centrality Potential.}
To ensure that the updates remain close to the central path, i.e., \(s + \mu \mW(x) \nabla \Phi(x) \approx 0\), a centrality potential function is used. This potential measures the proximity of the current point to the true central path and helps guide the algorithm towards the next point \((x, s, \mu)\). By minimizing this potential at each iteration, the algorithm ensures steady progress toward the optimal solution.

Formally, \cite{brand2021minimumcostflowsmdps} define this centrality potential as:
\[
\Psi(x, s, \mu) = \sum_{i=1}^{m} \cosh\left(\lambda \left(\frac{s_i + \mu w(x)_i \phi_i'(x_i)}{\mu w(x)_i \sqrt{\phi_i''(x_i)}}\right)\right),
\]
where \(\lambda = \Theta(\log m / \varepsilon)\).

Using \(\Psi\), a gradient vector \(g\) is computed, which is then used to update \(x\) and \(s\). The update process also leverages an orthogonal projection matrix \(\mP\), defined as:
\[
\mP = \Tau^{-1/2} \Phi''(x)^{-1/2} \mA \left(\mA^\top \Tau^{-1} \Phi''(x)^{-1} \mA\right)^{-1} \mA^\top \Phi''(x)^{-1/2} \Tau^{-1/2}.
\]
where $\Tau = \diag(\tau)$ with $\tau = w(\phi''(x)^{-/2})$, and $\Phi''(x) = \diag(\phi''(x))$. The variables \(x\), \(s\), and \(\mu\) are then updated roughly as follows:
\[
s^\new \leftarrow s + \Tau^{1/2} \Phi''(x)^{1/2} \mP \Tau^{1/2},
\]
\[
x^\new \leftarrow x + \Tau^{1/2} \Phi''(x)^{1/2} (\mI - \mP) \Tau^{1/2},
\]
\[
\mu^\new \leftarrow (1 - r) \mu,
\]
where \(r \in \tO(n^{-1/2})\) is a parameter chosen to control the step size.

Notably, the algorithm does not compute the exact value of $\left(\mA^\top \Tau^{-1} \Phi''(x)^{-1} \mA\right)^{-1}$. Instead, it approximates this matrix using a spectral approximation, obtaining $\mH \approx \left(\mA^\top \Tau^{-1} \Phi''(x)^{-1} \mA\right)$. This approximation is critical for ensuring the computational efficiency of the algorithm. Further details of this process are discussed in \Cref{subsection:pathfollowing}.

\subsection{Main Bottlenecks \& Tools to Overcome Them} \label{subsection:bottlenecks}

As discussed in the previous subsection, the main computational challenge in adapting the algorithm by \cite{brand2021minimumcostflowsmdps} lies in the path-following algorithm. A closer look at this algorithm reveals that, while other steps are non-trivial, two specific parts become the main bottlenecks when converting to the two-party communication model:
\begin{enumerate}
    \item \emph{Computing \(\ell_p\)-Lewis Weights}: This involves solving the equation 
    \[
    w(x) = \sigma\left(\mW(x)^{\frac{1}{2} - \frac{1}{p}} \left(\nabla^2 \Phi(x)\right)^{-\frac{1}{2}} \mA\right).
    \]
    \item \emph{Computing Spectral Approximation}: This can be formally described as finding a matrix \(\mB \in \R^{\wt{n} \times n}\) such that \(\mB^\top \mB \approx \mA^\top \mA\) for a given matrix $\mA \in \R^{m \times n}$.
\end{enumerate}

In their work, \cite{ghadiri2024improvingbitcomplexitycommunication} study and improve the bit complexity of solving some fundamental convex optimization problems in communication models that are very similar to the two-party communication model. One of these problems is solving general linear programs with one-sided constraints (\(x \geq 0\) instead of \(\ell \leq x \leq u\)). They use the algorithm from \cite{tallDense} as a template and show that such LPs can be solved using \(\tO(n^2)\) bits of communication.

Our goal is to use some of the tools from \cite{ghadiri2024improvingbitcomplexitycommunication} and apply them to the IPM of \cite{brand2021minimumcostflowsmdps} (instead of the one from \cite{tallDense}, which \cite{ghadiri2024improvingbitcomplexitycommunication} originally use) in order to overcome the bottlenecks. In the following, we discuss these necessary tools.

\paragraph{$\ell_p$-Lewis Weights.} 
In the communication setting, we can efficiently compute $\ell_p$-Lewis weights by adapting Lemma 4.7 of \cite{ghadiri2024improvingbitcomplexitycommunication}. This allows us to obtain approximate Lewis weights with near-optimal bit complexity, while ensuring that Alice and Bob each hold the relevant parts of the output.

\paragraph{Spectral Approximation.} 
Using leverage score estimates and sampling techniques from \cite{ghadiri2024improvingbitcomplexitycommunication}, we can compute and maintain spectral approximations with nearly linear communication cost. While the details are technical, the key takeaway is that spectral approximation can be performed within \(\tO(nkL)\) bits per iteration, and maintained more efficiently across iterations, where \(L\) is the bit complexity of each entry, \(k\) is the maximum number of non-zero entries per row of the constraint matrix, and \(n\) is the number of variables.

In \Cref{subsec:lewisWeightsAndSpecApprox}, we will discuss these points in more detail.

\subsection{Goals of This Work} \label{subsection:maingoals}
Here, we summarize the goals of this work. We:
\begin{itemize}
    \item \emph{Extend the techniques by \cite{ghadiri2024improvingbitcomplexitycommunication}:} We refine and complete the techniques introduced by \cite{ghadiri2024improvingbitcomplexitycommunication}, particularly for the computation of spectral approximations. In \Cref{algo:SpecApprox}, we present a method to construct spectral approximations using leverage scores, and in \Cref{lemma:specApproxlemma}, we analyze the bit complexity of this approach.
    
    \item \emph{Simplify the sequential algorithm for the communication model:} We modify the original algorithm (\Cref{algo:lsstep}) to take advantage of the two-party communication model. Specifically, since maintaining approximations is unnecessary in this setting (because of infinite internal computational power), we use exact values directly, simplifying the overall algorithm.

    \item \emph{Convert the sequential algorithm to the two-party communication model:} Using the tools discussed in the following section, we describe in \Cref{algo:twoplsstep} how to perform path-following in the two-party communication setting. In \Cref{subsection:proofIPM}, we provide a detailed bit complexity analysis of each step in our algorithm, formally proving \Cref{theorem:generalLP}.

    \item \emph{Extend the results to MinCost Flow and MaxFlow:} Building on our result for general LPs with two-sided constraints, we develop algorithms for the minimum cost flow problem and the maximum flow problem (\Cref{section:MCMFinTPCModel}).

\end{itemize}

\clearpage

\section{Communication Complexity of Linear Programming}\label{section:LPinCC}

In this section, we bound the communication complexity of solving linear programs with two-sided constraints, formulated as:

\begin{align}
    \min_{\substack{\mA^\top x = b \\ \ell \le x \le u}} c^\top x, \label{eq:generalLP}
\end{align}

where $\mA \in \R^{m \times n}$, $c \in \R^m$, $b \in \R^n$, and $\ell, u \in \R^m$ define the lower and upper bounds for $x \in \R^m$, respectively. As discussed in the introduction, we assume that Alice and Bob each hold distinct rows of \(\mA\), referred to as \(\mA\pA\) and \(\mA\pB\), along with their corresponding portions of \(c\), \(\ell\), and \(u\). Both parties share knowledge of \(b\). Additionally, as stated earlier, the bit complexity of each entry in \(\mA\), \(b\), and \(c\) is assumed to be bounded by \(L\). The main result of this section is as follows:

\theoremgeneralLP*

The algorithm referenced in \Cref{theorem:generalLP} is essentially an interior-point method (IPM), introduced by \cite{brand2021minimumcostflowsmdps}, which has been modified using techniques from \cite{ghadiri2024improvingbitcomplexitycommunication}. 

In the following subsections, we aim to prove \Cref{theorem:generalLP}. We begin with several lemmas and techniques, mostly based on \cite{ghadiri2024improvingbitcomplexitycommunication}, which address the main bottlenecks discussed in \Cref{subsection:bottlenecks}. We then explain the LP-solving IPM of \cite{brand2021minimumcostflowsmdps} and show how to adapt it to the two-party communication setting. In particular, we first focus on the path-following algorithm, the most challenging part, and then verify how to find an initial feasible point.

The IPM itself is quite complex. In essence, the approach of \cite{brand2021minimumcostflowsmdps} for solving LPs with two-sided constraints is to (1) modify the LP so it has a trivial feasible point, (2) apply the path-following algorithm to the modified LP, and (3) derive a solution to the original LP. To simplify the exposition, they assume the path-following algorithm runs directly on the original LP and later prove this makes no difference. We follow the same convention and, without loss of generality, assume the algorithm operates on the original LP.

\textbf{Note for the reader:} the most relevant part of this section is the collection of lemmas and techniques in the next subsection. These are adapted from \cite{ghadiri2024improvingbitcomplexitycommunication}, and some do not appear explicitly in their work. The remaining discussion mainly serves to verify the correctness of the theorem and the protocol.

\subsection{$\ell_p$-Lewis Weights and Spectral Approximation in Communication Setting} \label{subsec:lewisWeightsAndSpecApprox}

This subsection is mostly based on the results and lemmas of \cite{ghadiri2024improvingbitcomplexitycommunication}. While they present these results for the coordinator model, they can be directly extended to the communication model. 

The coordinator model is defined as a model with a central coordinator who performs the common computations and can communicate with each party. The bit complexity is measured as the total number of bits exchanged between the coordinator and the parties. In our work, we also assume the presence of a coordinator, so that Alice and Bob have symmetric roles. The coordinator does not have any knowledge of the graph and is only responsible for handling the common computations. In the original two-party communication model, one of the parties could simply take the role of the coordinator.

First, let us discuss how to compute \(\ell_p\)-Lewis weights efficiently. Throughout this section, we assume that the bit complexities of \(\mA\), \(b\), \(c\), \(u\), and \(\ell\) are bounded by \(L\).

\begin{lemma}[Adapted from lemma 4.7 of \cite{ghadiri2024improvingbitcomplexitycommunication}]\label{lemma:lewis-weight-communication}
Consider the two-party communication setting with a central coordinator. Let \(1 > \eta > 0\), \(0 < \epsilon < 0.5\), and \(p \in (0, 4)\). Suppose the input matrix \(\mA = [\mA^{(i)}] \in \R^{m \times n}\) is distributed across Alice and Bob. There exists a randomized algorithm that, with high probability, outputs a vector \(\hw\) such that
\[
\hw \approx_{\epsilon} \sigma(\hat{\mW}^{1/2 - 1/p} \mA) + \eta \cdot \boldsymbol{1}.
\]
The total communication cost of this algorithm, in terms of bits, is given by:
\[
\tO\left((n k L + n(L + \log \kappa + p^{-1} \log(\eta^{-1}))) \cdot \frac{\log(\eta^{-1}\varepsilon^{-1} p^{-1})}{1 - |p/2 - 1|}\right),
\]
where each row or $\mA$ has at most $k$ non-zero entries. Additionally, Alice and Bob will each hold their respective portions of \(\hw\), corresponding to the rows of \(\mA\).
\end{lemma}

Furthermore, for the IPM to take the short steps, it is necessary to efficiently compute a spectral approximation of a matrix in the two-party communication setting. To achieve this, \cite{ghadiri2024improvingbitcomplexitycommunication} uses an approximation of the matrix's leverage scores to sample its rows, which are then used to form a spectral approximation of the original matrix. Note that \cite{ghadiri2024improvingbitcomplexitycommunication} defines that for $\lambda > 1$, a matrix $\tilde{\mA} \in \R^{n' \times d}$ is a $\lambda$-spectral approximation of $\mA \in \R^{n \times d}$ if
\[
\frac{1}{\lambda} \mA^\top \mA \preceq \tilde{\mA}^\top \tilde{\mA} \preceq \mA^\top \mA
\]
(e.g. See \cite{ghadiri2024improvingbitcomplexitycommunication}, definition 1.21). This definition is equivalent to the one we have used thus far. Therefore, for the sake of convenience, we will adopt this notation for the remainder of this section.

\begin{lemma}[Approximating Leverage Scores, Adapted from lemma 3.2 of \cite{ghadiri2024improvingbitcomplexitycommunication}]\label{lemma:levScore}
Consider the two-party communication setting with a central coordinator with matrix $\mA=[\mA^{(i)}]\in \R^{m \times n}$, where each row of $\mA$ has at most $k$ non-zero entries, and $m \geq 5$, there is a randomized algorithm that, with \[\Otil(n k L + n(L + \log(\kappa))) \text{ bits of communication}\] and, with high probability, computes a vector $\hat{\sigma}\in \R^m$ such that $\norm{\hat{\sigma}}_{1} \leq 9n$ and $\hat{\sigma}_i\geq\sigma_i(\mA)$, for all $i \in [m]$, where $\sigma_i(\mA)$ is the $i^\mathrm{th}$ leverage score of the matrix $\mA$. Each entry of $\hat{\sigma}$ is stored on the machine that contains the corresponding row.
\end{lemma}

\cite{ghadiri2024improvingbitcomplexitycommunication} then shows that the leverage score overestimates, $\hat{\sigma}$, can be used in combination with the sampling function in \Cref{def:sample} to compute a spectral approximation. \Cref{algo:SpecApprox} demonstrates this procedure.

\begin{definition}[Sampling Function, Definition 3.3 of \cite{ghadiri2024improvingbitcomplexitycommunication}]\label{def:sample}
Given a vector $u \in \R^n_{\geq 0}$, a parameter $\alpha> 0$, and a positive constant $c$, we define vector $p\in \R^n_{\geq 0}$ as $p_i = \min(1, \alpha c \log n \cdot u_i)$. We  define the function $\textsc{Sample}(u, \alpha, c)$ to be one which returns a random diagonal $n\times n$ matrix $\mS$ with independently chosen entries: 
\[\mS_{ii} = \twopartdef{\frac{1}{\sqrt{p_i}}}{\textrm{with probability } p_i}{0}{\textrm{otherwise}}.\]
\end{definition}

\begin{algorithm2e}[ht]
\caption{Protocol for computing a $\lambda$-spectral approximation in the two-party communication setting with a central coordinator \label{algo:SpecApprox}}
\SetKwInOut{Input}{Input}
\SetKwInOut{Output}{Output}
\Input{Matrix $\mA \in \R^{m \times n}$ with Alice and Bob each holding a subset of the rows of $\mA$, and the probability parameter $c$.}
\Output{Spectral approximation $\tilde{\mA}$, stored in the coordinator.}
\SetKwProg{Procedure}{Procedure}{}{}
\Procedure{\SpecApprox$(\mA, c)$}{
    Alice and Bob compute the leverage score overestimates $\hat{\sigma} \in \R^n_{\geq 0}$ according to \Cref{lemma:levScore}, with each storing the portion corresponding to their respective subset of rows of $\mA$. \label{line:specapprox:levScore} \\

    Using $\hat{\sigma}$, Alice and Bob form the diagonal sampling matrix $\mS = \textsc{Sample}(\hat{\sigma}, \alpha, c) \in \R^{n \times n}$ according to \Cref{def:sample}, where $\alpha = (\frac{\lambda + 1}{\lambda - 1})^2$. \label{line:specapprox:sample} \\

    Alice and Bob send the non-zero entries of $\mS$, along with their corresponding rows of $\mA$, to the coordinator. \label{line:specapprox:sendS}\\

    The coordinator constructs the matrices $\tilde{\mS} \in \R^{\tilde{n} \times \tilde{n}}$ and $\tilde{\mA} = \sqrt{\frac{\lambda + 1}{2 \lambda}} \tilde{\mS} \mA \in \R^{\tilde{n} \times n}$, where $\tilde{n} \in \tO(n)$. \label{line:specapprox:computeSpec}
}
\end{algorithm2e}

\Cref{algo:SpecApprox} utilizes the leverage score overestimates $\hat{\sigma}$ to construct a matrix $\tilde{\mA} \in \R^{\tilde{n} \times n}$, composed of $\tilde{n} \in \tO(n)$ rows of $\mA$ and the sampling matrix $\mS$. Each row $i$ is sampled independently with probability $p_i \propto \alpha \hat{\sigma}_i$, for some sampling rate $\alpha > 0$. With a constant $\alpha$, the number of rows in $\tilde{\mA}$, proportional to $\alpha \cdot \norm{\hat{\sigma}}_1$, is with high probability $\tO(n)$. This ensures that the number of rows and entries communicated is only $\tO(n)$, resulting in an overall bit complexity of $\tO(nkL)$, assuming each row has at most $k$ non-zero entries. This guarantee is formalized in \Cref{lemma:specApproxlemma}.

\begin{lemma}[Spectral Approximation via Leverage Score, partially adopted from lemma 3.4 of \cite{ghadiri2024improvingbitcomplexitycommunication}] \label{lemma:specApproxlemma} 
Given a matrix $\mA \in \R^{m \times n}$, a sampling rate $\alpha > 1$, and a fixed constant $c > 0$. Let  $\sigma \in \R^m_{\geq 0}$ be a vector of leverage score overestimates, that is, \[\sigma \geq \sigma(\mA),\,\, \text{ and } \mS := \texttt{Sample}(\sigma, \alpha, c)\] as in \Cref{def:sample}. Then, with probability at least $1-n^{-c/3} - (3/4)^n$, the following results hold: 
\[\text{nnz}(\mS) = 2c \alpha \norm{\sigma}_1 \log n \text{ and } \tilde{\mA} = \frac{1}{\sqrt{1+\alpha^{-1/2}}} \mS \mA \approx_{\left(\frac{1+\alpha^{-1/2}}{1-\alpha^{-1/2}}\right)}\mA.\]

Computing $\tilde{\mA}$ could be done using \Cref{algo:SpecApprox} and requires 
\begin{align*}
    \tO(nkL + n \log \kappa) \text{ bits of communication}.
\end{align*}
\end{lemma}

\begin{proof}[Proof of \Cref{lemma:specApproxlemma}]

Correctness follows from lemma 3.4 in \cite{ghadiri2024improvingbitcomplexitycommunication}. To analyze the communication complexity, observe that the only communication-intensive steps in \Cref{algo:SpecApprox} occur in \Crefrange{line:specapprox:levScore}{line:specapprox:sendS}. By \Cref{lemma:levScore}, \Cref{line:specapprox:levScore} requires just \(\tO(nkL + n \log \kappa)\) bits. In \Cref{line:specapprox:sendS}, Alice and Bob communicate the non-zero entries of \(\mS\). Since \(\text{nnz}(\mS) = 2c \alpha \|\sigma\|_1 \log n\), this can be accomplished with a total of \(\tO(nkL)\) bits. Thus, the total communication complexity is \(\tO(nkL + n \log \kappa)\) bits.

\end{proof}

Using \Cref{algo:SpecApprox}, we can compute \(\mH\), which spectrally approximates \(\mA^\top{\sigma}^{-1}\Phi''({x})^{-1}\mA\), with a communication cost of \(\tO(nkL)\) bits per iteration. This leads to an overall bit complexity of \(\tO(n^{1.5}kL)\) bits of communication. While this approach is efficient for sufficiently sparse matrices (e.g., in the case of minimum cost maximum flow), the specific step of computing the spectral approximation in every iteration becomes inefficient for dense matrices.  

To address this inefficiency, it is possible to compute the spectral approximation once and maintain it across subsequent iterations, instead of recomputing it in each iteration. This idea aligns with the algorithm introduced by \cite{ghadiri2024improvingbitcomplexitycommunication}, which maintains a spectral approximation in a coordinator. In this work, we utilize their inverse maintenance method as a black box.

\begin{lemma}[Inverse Maintenance, Algorithm 10 of \cite{ghadiri2024improvingbitcomplexitycommunication}]
\label{lemma:inversemaintenance}
Given a matrix \(\mA \in \R^{m \times n}\) and diagonal matrices \(\mD^{(0)}, \mD^{(1)}, \dots, \mD^{(r)} \in \R^{m \times m}\), where each matrix \(\mD^{(i)}\) becomes available after the output for the input \(\mD^{(i-1)}\) has been returned, there exists a randomized algorithm that, for \(i \geq 0\), computes the spectral approximation of \((\mA^\top \mD^{(i)} \mA)^{-1}\) using \(\tO(r^2 kL \log^2 \varepsilon^{-1})\) bits of communication.
\end{lemma}

While this procedure addresses one of the main bottlenecks in the computation, we emphasize that it does not significantly improve the result of \Cref{theorem:generalLP}. Rather, it serves as a step toward achieving an overall bit complexity of \(\tO(n^{1.5} + nk)\), compared to the current \(\tO(n^{1.5}k)\), by improving the efficiency of spectral approximation computation and maintenance.

\subsection{Path Following Algorithm} \label{subsection:pathfollowing}

In this section, we present an algorithm for solving general linear programs (LPs) within a two-party communication model. We begin by introducing a modified version of the Interior Point Method (IPM) from \cite{brand2021minimumcostflowsmdps}, which was initially developed for sequential computation. The algorithm starts from an initial feasible point located on or near the central path and proceeds by making short steps over \(\tO(\sqrt{n})\) iterations using the Lee-Sidford barrier. We then detail how this modified IPM can be adapted to the two-party communication model using the techniques introduced in \cite{ghadiri2024improvingbitcomplexitycommunication}.

During the interior point method (IPM), we maintain and update triples \((x, s, \mu)\). At each step, starting from the current point \((x, s, \mu)\), we introduce a \emph{weight function} \(\tau: \R^m \rightarrow \R^m_{>0}\) that governs the central path. Conceptually, \(\tau(x)_i\) reflects the weight assigned to the \(i\)-th barrier function \(\phi_i\). To construct the central path, we use a regularized Lewis weight.

\begin{definition}[Regularized Lewis Weights for a Matrix]
\label{def:matrixlewis}
\label{def:v}
Let \(p = 1 - \frac{1}{4\log(4m/n)}\) and \(v \in \R^m_{>0}\) be a vector such that \(v_i \geq n/m\) for all \(i\) and \(\|v\|_1 \leq 4n\). For a given matrix \(\mA\), the \emph{($v$-regularized) \(\ell_p\)-Lewis weights} \(w(\mA) \in \R^m_{>0}\) are defined as the solution to:
\[
w(\mA) = \sigma( \mW^{\frac{1}{2} - \frac{1}{p}} \mA ) + v
\]
where \(\mW = \diag(w(\mA))\).
\end{definition}

As noted by \cite{brand2021minimumcostflowsmdps}, the regularization vector \(v\) should be chosen such that the weight at each coordinate is at least \(n/m\). To satisfy this condition, we set \(v = \frac{n}{m} \mathbf{1}\) throughout this work.

\begin{definition}[Regularized Lewis Weights for \(c\)]
\label{def:lewis}
Let \(p = 1 - \frac{1}{4\log(4m/n)}\). Given vectors \(c\) and \(v \in \R^m_{>0}\), the \emph{($v$-regularized) \(\ell_p\)-Lewis weights} \(w(c): \R^m_{>0} \rightarrow \R^m_{>0}\) are defined as \(w(c) = w(\mC\mA)\), following the definition from \Cref{def:matrixlewis}.
\end{definition}

Note that the previous definition assumes the constraint matrix \(\mA\) of the LP can be omitted from the context.

\begin{definition}[Central Path Weights]
\label{def:cpweights}
The \emph{central path weights} are defined as \(\tau(x) = w(\phi''(x)^{-\frac{1}{2}})\), using a fixed regularization vector \(v\).
\end{definition}

Throughout the iterations, the algorithm ensures that the points \((x, s, \mu)\) maintain a centrality condition, as defined below:

\begin{definition}[$\eps$-Centered Point]
\label{def:centered}
We define \((x, s, \mu) \in \R^m \times \R^m \times \R^m_{>0}\) as being $\eps$-centered for \(\eps \in (0, 1/80]\) if the following conditions are satisfied, where \(\cnorm = \frac{C}{1-p}\) for a constant \(C \ge 100\):
\begin{enumerate}
    \item (Approximate Centrality): The centrality condition holds approximately: \[
    \left\| \frac{s + \mu \tau(x) \phi'(x)}{\mu \tau(x) \sqrt{\phi''(x)}} \right\|_\infty \le \eps.
    \]
    \item (Dual Feasibility): There exists a vector \(z \in \R^n\) such that \(\ma z + s = c\), ensuring dual feasibility.
    \item (Approximate Primal Feasibility): The primal feasibility condition is satisfied up to an error bound: \[
    \| \ma^\top x - b \|_{(\mA^\top (\Tau(x) \Phi''(x))^{-1} \mA)^{-1}} \le \frac{\eps \gamma}{\cnorm}.
    \]
\end{enumerate}
\end{definition}

\cite{brand2021minimumcostflowsmdps} introduce a randomly sampled diagonal scaling matrix \(\mR\) in their algorithm, and it is essential that this matrix satisfies certain properties to ensure the progress of the IPM. Below, the required key properties are outlined. This definition encompasses distributions like independently sampling each coordinate as a Bernoulli variable with probabilities \(p_i\), or summing multiple samples weighted by \(p_i\).

\begin{definition}[Valid Sampling Distribution]
\label{def:validdistro}
Given vector $\d_r, \mA, {x}, {\tau}$ as in \ShortStep~(\Cref{algo:lsstep}), we say that a random diagonal matrix $\mR \in \R^{m \times m}$ is $\cvalid$-\emph{valid} if it satisfies the following properties, for $\bar{\mA} = {\Tau}^{-\frac{1}{2} }\Phi''({x})^{-\frac{1}{2} }\mA$. We assume that $\cvalid \ge \cnorm$.
\begin{itemize}
\item(Expectation) We have that $\E[\mR] = \mI$.
\item(Variance) For all $i \in [m]$, we have that $\Var[\mR_{ii}(\d_r)_i] \le \frac{\gamma|(\d_r)_i|}{\cvalid^2}$.
and $\E[\mR_{ii}^2] \le 2\sigma(\bar{\mA})_i^{-1}$.
\item(Covariance) For all $i \neq j$, we have that $\E[\mR_{ii}\mR_{jj}] \le 2$.
\item(Maximum) With probability at least $1-n^{-10}$ we have that $\|\mR\d_r-\d_r\|_\infty \le \frac{\gamma}{\cvalid^2}$.
\item(Matrix approximation) We have that $\bar{\mA}^\top\mR\bar{\mA} \approx_\gamma \bar{\mA}^\top\bar{\mA}$ with probability at least $1-n^{-10}$.
\end{itemize}
\end{definition}

With the previous notations and definitions established, we can now describe the IPM. \Cref{algo:pathfollowing} (Algorithm 2 in \cite{brand2021minimumcostflowsmdps}) leverages \Cref{algo:lsstep} (a modified version of Algorithm 1 from \cite{brand2021minimumcostflowsmdps}) to solve the linear program through a series of short steps. This framework allows us to solve linear programs in the sequential setting, assuming the availability of an initial feasible point.

\begin{algorithm2e}[ht]
\caption{Path Following Meta-Algorithm for solving $\min_{\mA^\top x = b, \ell \le x \le u } c^\top x,$ given an initial point $\eps/\cstart$-centered point $(x^\init, s^\init, \mu)$ for large $\cstart$. \label{algo:pathfollowing}}
\SetKwProg{Proc}{procedure}{}{}
\Proc{\PathFollowing$(\mA, \ell, u, \mu, \mu^\final)$}{
	Define $r$ as in \Cref{algo:lsstep}. \\
	\While{$\mu > \mu^\final$}{
		$(x^\new, s^\new) \assign \ShortStep(x, s, \mu, (1-r)\mu)$. \\
		$x \assign x^\new, s \assign s^\new, \mu \assign (1-r)\mu$. \\
	}
	Use \Cref{lemma:finalpoint} to return a point $(x^\final, s^\final)$.
}
\end{algorithm2e}

\begin{algorithm2e}[ht]
\caption{Short Step (Lee Sidford Barrier) \label{algo:lsstep}}
\SetKwProg{Proc}{procedure}{}{}
\Proc{\ShortStep$(x,s,\mu,\mu^\new)$}{
	Fix $\tau(x) = w(\phi''(x)^{-\frac{1}{2}})$. \\
	Let $\alpha = \frac{1}{4\log(4m/n)}, \eps = \frac{\alpha}{C}, \lambda = \frac{C\log(Cm/\eps^2)}{\eps}, \gamma = \frac{\eps}{C\lambda}, r = \frac{\eps\gamma}{\cnorm\sqrt{n}}$. \\
	Assume that $(x,s,\mu)$ is $\eps$-centered 
		and $\d_\mu \defeq \mu^\new-\mu$ satisfies $|\d_\mu| \le r\mu$. 
		\label{line:ipm:centered}\\
	Set $y \in \R^m$, so that $y_i = \frac{s_i+\mu\tau(x)_i\phi_i'(x_i)}{\mu\tau(x)_i\sqrt{\phi''_i(x_i)}}$ for all $i \in [m]$. 
			\label{line:ipm:y}\\
    Let $g = -\gamma\g\Psi({y})$, 
			where $\Psi(y) = \sum_{i = 1}^m cosh(\lambda y_i)$. 
			\label{line:ipm:g}\\
	Let $\mH \approx_\gamma \bar{\mA}^\top \bar{\mA} = \mA^\top{\Tau}^{-1}\Phi''({x})^{-1}\mA$, where $\bar{\mA} = {\Tau}^{-\frac{1}{2}}\Phi''({x})^{-\frac{1}{2}}\mA$.
			\label{line:ipm:H}\\
	Let $\d_1 = {\Tau}^{-1}\Phi''({x})^{-\frac{1}{2} }\mA\mH^{-1}\mA^\top \Phi''({x})^{-\frac{1}{2} }g$ 
			and $\d_2 = {\Tau}^{-1}\Phi''({x})^{-\frac{1}{2} }\mA\mH^{-1}(\mA^\top x - b)$. \\ 
			Let  $\d_r = \d_1 + \d_2$. 
			\label{line:ipm:delta_r}\\
	Let $\mR \in \R^{m\times m}$ be a $\cvalid$-valid random diagonal matrix 
			for large $\cvalid$. 
			\Comment{\Cref{def:validdistro}} 
			\label{line:ipm:R}\\
	${\d}_x \leftarrow \Phi''({x})^{ -\frac{1}{2} }\left( g - \mR \d_r \right).$ 
			\label{line:ipm:delta_x}\\
	${\d}_s \leftarrow \mu {\Tau}\Phi''({x})^{ \frac{1}{2} } \d_1$. 
			\label{line:ipm:delta_s}\\
	$x^\new \leftarrow x + {\d}_x$ 
			and $s^\new \leftarrow s + {\d}_s.$ 
			\label{line:ipm:xnewsnew}\\
	\Return~$(x^\new,s^\new)$.
}
\end{algorithm2e}

In \cite{brand2021minimumcostflowsmdps}, it is shown that \Cref{algo:pathfollowing} requires \(\tO(\sqrt{n}\log(\mu/\mu^\final))\) iterations to reach the final solution:

\begin{lemma}[Lemma 4.12 of \cite{brand2021minimumcostflowsmdps}]\label{lemma:pathfollowing} 
 
\Cref{algo:pathfollowing} makes \(\tO(\sqrt{n}\log(\mu/\mu^\final))\) calls to \Cref{algo:lsstep}, and with probability at least \(1 - m^{-5}\), the following conditions are satisfied at the beginning and end of each call to \ShortStep:
\begin{enumerate}
\item \emph{Slack feasibility}: There exists a vector \(z \in \R^n\) such that \(s = \mA z + c\).
\item \emph{Approximate feasibility}: \(\|\mA^\top x - b\|_{(\mA^\top(\Tau(x)\Phi''(x))^{-1}\mA)^{-1}} \le \eps \gamma / \cnorm\).
\item \emph{Potential function}: The expected value of the potential function satisfies \(\E[\Psi(x, s, \mu)] \le m^2\), where the expectation is taken over the randomness in \(x\) and \(s\).
\item \emph{\(\eps\)-centered}: The tuple \((x, s, \mu)\) is \(\eps\)-centered.
\end{enumerate}
Moreover, for \(\mu^\final \le \d / (Cn)\), the final output \(x^\final\) satisfies \(\mA x^\final = b\) and \(c^\top x^\final \le \min_{\substack{\mA^\top x = b \\ \ell_i \le x_i \le u_i}} c^\top x + \d.\)
\end{lemma}
    
\begin{lemma}[Lemma 4.11 of \cite{brand2021minimumcostflowsmdps}]\label{lemma:finalpoint} 
Given an $\eps$-centered
point $(x,s,\mu)$ where $\eps\le1/80$, using a Laplacian solver, we can compute a point $(x^{\final},s^{\final})$
satisfying 
\begin{enumerate}
	\item $\mA^{\top}x^{\final}=b$, $s^{\final}=\mA y+c$ for some $y$. 
	\item $c^{\top}x^{\final}-\min_{\substack{\mA^{\top}x=b\\
			\ell_i \le x_i \le u_i \forall i
		}
	}c^{\top}x\ls n\mu.$ 
\end{enumerate}
\end{lemma}

The last step of \Cref{algo:pathfollowing}, as noted in \Cref{lemma:finalpoint}, uses a Laplacian solver, which for the system \( (\mA^\top \mD \mA)x = b \), if \( x \in \mathbb{R}^n \) exists, w.h.p. returns an approximation \( \overline{x} \in \mathbb{R}^n \), such that \( \|\overline{x} - x\|_{\mA^{\top}\mD\mA} \leq \varepsilon \|x\|_{\mA^{\top}\mD\mA} \). In the following, we show that this is equivalent to solving \( (\mA^\top \mD \mA)x = b \) using a spectral approximation of \( \mD^{\frac{1}{2}}\mA \), which provides \( \mH \in \mathbb{R}^{n \times n} \) such that \( \mH \approx \mA^\top \mD \mA \).

\begin{lemma} \label{lemma:specApproxEquivalency}
Let $\mB, \mH \in \R^{n \times n}$ be symmetric positive definite matrices so that $\mH \approx_\lambda \mB$ for $\lambda > 0$, i.e., $\exp(-\lambda)\mB \preceq \mH \preceq \exp(\lambda)\mB$. Furthermore, let $x = \mB^{-1} b \in \R^n$ and $\overline{x} = \mH^{-1} b \in \R^n$. Then \( \|\overline{x} - x\|_{\mB} \leq \varepsilon \|x\|_{\mB} \) for $\varepsilon = \left(e^{\lambda} (e^{\lambda} - 1)\right)^2$.
\end{lemma}

\begin{proof}
Define $\d = \mB (\overline{x} - x)$ and $\Delta = \mH - \mB$. So we have
\begin{align*}
(\mB + \Delta)^{-1} b
= ~ & \mH^{-1} b \\
= ~ & \overline{x} \\
= ~ & x + \mB^{-1} \d \\
= ~ & \mB^{-1} (b + \d).
\end{align*}
Hence, we have
\begin{align*}
b 
= ~ & (\mB + \Delta) \mB^{-1} (b + \d) \\
= ~ & (\mI + \Delta \mB^{-1}) (b + \d) \\
= ~ & b + \Delta \mB^{-1} b + (\mI + \Delta \mB^{-1}) \d.
\end{align*}
Letting $\mS = (\mI + \Delta \mB^{-1})^{-1} \Delta \mB^{-1}$, this results in
\begin{align} \label{eq:deltaSb}
\d = - (\mI + \Delta \mB^{-1})^{-1} \Delta \mB^{-1} b = - \mS b.
\end{align}
On the other hand, since $e^{-\lambda} \mB \preceq \mH = \mB + \Delta \preceq e^{\lambda} \mB$ and $\mB$ as well as $\mI + \Delta \mB^{-1} = \mH$ are positive definite, we have
\begin{align*}
(e^{-\lambda} - 1) \mI \preceq \Delta \mB^{-1} \preceq (e^{\lambda} - 1) \mI \\
e^{-\lambda} \mI \preceq (\mI + \Delta \mB^{-1})^{-1} \preceq e^{\lambda} \mI.
\end{align*}
Combining the equations above gives us 
\begin{align*}
e^{-\lambda} (e^{-\lambda} - 1) \mI \preceq \mS = (\mI + \Delta \mB^{-1})^{-1} \Delta \mB^{-1} \preceq e^{\lambda} (e^{\lambda} - 1) \mI.
\end{align*}
Now, since $\mS$, as a product of symmetric and positive definite matrices, is positive definite, we have
\begin{align*}
\mS^\top \mB^{-1} \mS 
= & ~ \mS \mB^{-1} \mS \\
\preceq & ~ (e^{\lambda} (e^{\lambda} - 1) \mI) \mB^{-1} (e^{\lambda} (e^{\lambda} - 1) \mI) \\
= & ~ (e^{\lambda} (e^{\lambda} - 1))^2 \mB^{-1} \\
= & ~ \varepsilon \mB^{-1}.
\end{align*}
This means that $\varepsilon \mB^{-1} - \mS^\top \mB^{-1} \mS$ is positive semi-definite. Thus, we have
\begin{align} \label{eq:specapproxdiff}
b^\top (\varepsilon \mB^{-1} - \mS^\top \mB^{-1} \mS) b = \varepsilon b^\top \mB^{-1} b - b^\top \mS^\top \mB^{-1} \mS b \geq 0.
\end{align}
Combining equations \eqref{eq:deltaSb} and \eqref{eq:specapproxdiff}, we have
\begin{align*}
\|\overline{x} - x\|_{\mB} 
= & ~ \| \mB^{-1} \d \|_\mB \\
= & ~ \d^\top \mB^{-1} \d \\
= & ~ b^\top \mS^\top \mB^{-1} \mS b \\
\leq & ~ \varepsilon b^\top \mB^{-1} b \\
= & ~ \varepsilon x^\top \mB x \\
= & ~ \varepsilon \|x\|_\mB.
\end{align*}
\end{proof}

In the following, in the \Cref{algo:twoplsstep} (page \pageref{algo:twoplsstep}), we present the two-party communication version of the \ShortStep~ algorithm. In this setup, Alice and Bob each hold a portion of the vectors \( x \) and \( s \). After each short step of the algorithm, both parties update their respective portions accordingly.

Similar to the sequential setting, \TwoPartyShortStep~ could be used in combination with \PathFollowing~ in order to solve the linear program \eqref{eq:generalLP} in the two-party communication setting.

\clearpage
\begin{algorithm2e}[ht]
\caption{Short Step (Lee-Sidford Barrier) -- Two-Party Communication with Central Coordinator \label{algo:twoplsstep}}
\SetKwInOut{In}{Input}
\SetKwInOut{Out}{Output}
\In{Vectors $x \coloneq [x^{(i)}] \in \R^m$ and $s \coloneq [s^{(i)}] \in \R^m$, as well as the constants $\mu$ and $\mu^\new$}
\Out{Vectors $x^\new \in \R^m$ and $s^\new \in \R^m$}
\SetKwProg{Proc}{Procedure}{}{}
\SetKwProg{SubProc}{Subprocedure}{}{}
\Proc{\TwoPartyShortStep$(x,s,\mu,\mu^\new)$}{
    Alice and Bob set $\tau(x) \approx_\varepsilon w(\phi''(x)^{-\frac{1}{2}})$ according to \Cref{lemma:lewis-weight-communication}. They each store the portion of $\tau(x)$ corresponding to the rows of $\mA$. \label{line:twoipm:lewWeight} \\
    
    The coordinator sets $\alpha = \frac{1}{4\log(4m/n)}, \varepsilon = \frac{\alpha}{C}, \lambda = \frac{C \log(Cm/\varepsilon^2)}{\varepsilon}, \gamma = \frac{\varepsilon}{C\lambda}, r = \frac{\varepsilon\gamma}{\cnorm \sqrt{n}}$. \label{line:ipm:const} \\
    
    Assume that $(x,s,\mu)$ is $\varepsilon$-centered and $\d_\mu \defeq \mu^\new-\mu$ satisfies $|\d_\mu| \le r\mu$. \label{line:twoipm:centered} \\
    
    Alice and Bob set $y \in \R^m$, where $y_i = \frac{s_i+\mu\tau(x)_i\phi_i'(x_i)}{\mu\tau(x)_i \sqrt{\phi''_i(x_i)}}$ for all $i \in [m]$. They each hold their portion of $y$ corresponding to rows of $\mA$. \label{line:twoipm:y}\\
 
    Alice and Bob set $g = -\gamma \nabla\Psi(y)$, where $\Psi(y) = \sum_{i = 1}^m \cosh(\lambda y_i)$. \label{line:twoipm:g}\\
    
    If a spectral approximation is not already available, using \Cref{algo:SpecApprox}, the coordinator computes $\mH \approx_\gamma \bar{\mA}^\top \bar{\mA} = \mA^\top{\Tau}^{-1}\Phi''(x)^{-1}\mA$, where $\bar{\mA} = {\Tau}^{-\frac{1}{2}}\Phi''(x)^{-\frac{1}{2}}\mA$. Otherwise the coordinator uses the result by inverse maintenance by \Cref{lemma:inversemaintenance}. In each iteration, if either $x_i$, $s_i$ or $\tau_i$ change for $i \in [m]$, we resample the $i^\mathrm{th}$ row according to its leverage scores and send it to the coordinator. The coordinator then updates the spectral approximation of $\mH^{-1}$ accordingly. \label{line:twoipm:H}\\
    
    With the coordinator's help, Alice and Bob compute $\d_1 = {\Tau}^{-1}\Phi''(x)^{-\frac{1}{2} }\mA\mH^{-1}\mA^\top \Phi''(x)^{-\frac{1}{2} }g$ and $\d_2 = {\Tau}^{-1}\Phi''(x)^{-\frac{1}{2} }\mA\mH^{-1}(\mA^\top x - b)$. Similar to previous steps, they each hold the portions of $\d_1$, $\d_2$, and $\d_r$ corresponding to the rows of $\mA$. \label{line:twoipm:subprocedure} \\
    
    \SubProc{Computing $\d_1$, $\d_2$, and $\d_r$}{
        Alice and Bob compute $v_1 = \mA^\top {\Phi''(x)^{-\frac{1}{2}}}g \in \R^n$ as well as $v_2 = \mA^\top x \in \R^n$ for their part of $\mA$ and $x$, and send them to the coordinator. \label{line:twoipm:subp1} \\

        The coordinator then computes $u_1 = \mH^{-1} {\mA}^\top {\Phi''(x)^{-\frac{1}{2}}} g = \mH^{-1} (v_1^{(A)} + v_1^{(B)}) \in \R^n$ as well as $u_2 = \mH^{-1} (\mA^\top x - b) = \mH^{-1} (v_2^{(A)} + v_2^{(B)} - b) \in \R^n$ and sends them to Alice and Bob. \label{line:ipm:subp2} \\

        Alice and Bob then multiply $u_1$ and $u_2$ by $\mD = {\Tau}^{-1}\Phi''(x)^{-\frac{1}{2}}\mA$ to compute $\d_1 = \mD u_1$ and $\d_2 = \mD u_2$. They set $\d_r = \d_1 + \d_2$. \label{line:twoipm:subp3}
    }
    
    Alice and Bob set $\mR \in \R^{m\times m}$ to be a $\cvalid$-valid random diagonal matrix for large $\cvalid$. They store the portion of $\mR$ corresponding to rows of $\mA$. Note that we assume that Alice and Bob share randomness. \label{line:twoipm:R} \\

    Alice and Bob set ${\d}_x = \Phi''(x)^{ -\frac{1}{2} }\left( g - \mR \d_r \right).$ 
            \label{line:twoipm:delta_x}\\
    Alice and Bob set ${\d}_s = \mu {\Tau}\Phi''(x)^{ \frac{1}{2} } \d_1$. 
            \label{line:twoipm:delta_s}\\
    Alice and Bob set $x^\tmp = x + {\d}_x$ 
            and $s^\tmp = s + {\d}_s.$ 
            \label{line:twoipm:xnewsnew}\\
}
\end{algorithm2e}
\clearpage

\subsection{Finding an Initial Feasible Point} \label{subsection:initialpoint}

\Cref{algo:pathfollowing} (and thus \Cref{algo:twoplsstep} as well), as mentioned in \Cref{lemma:pathfollowing}, return a final feasible \( \varepsilon \)-centered point \( (x^\target, s^\target, \mu^\target) \) given an initial feasible \( \varepsilon \)-centered point \( (x^\init, s^\init, \mu^\init) \). This final point can then be used, as noted in \Cref{lemma:finalpoint}, along with a Laplacian solver, to find a near-optimal, near-feasible solution to the linear program \eqref{eq:generalLP}. However, obtaining such an initial feasible \( \varepsilon \)-centered point is not straightforward. To address this issue, \cite{brand2021minimumcostflowsmdps} modify the linear program \eqref{eq:generalLP} to create a new linear program with a trivial feasible \( \varepsilon \)-centered point, demonstrating that a solution to the original linear program can be derived from this modified version. We adopt their approach in this section. Throughout this section, we let \( \d' = \frac{\delta}{10m2^{2L}} \) and \( \varepsilon = \frac{1}{4C \log (m/n)} \) for a sufficiently large constant \( C \).

\begin{definition}[Modified LP and its initial point, lemma 8.3 of \cite{brand2021minimumcostflowsmdps}] \label{lemma:modifiedLP}
Given a matrix \( \mA \in \R^{m \times n} \), a vector \( b \in \R^{n} \), a vector \( c \in \R^{m} \), an accuracy parameter \( \d \), and the following linear program:

\[
\min_{\substack{\mA^{\top} x = b \\ \ell \leq x \leq u}} c^{\top} x, 
\]

we define a new matrix \( \wt{\mA} \) as:

\[
\wt{\mA} \defeq \begin{bmatrix} \mA \\ \beta \mI_{n} \end{bmatrix},
\]

where \( \beta = \frac{\|b - \mA^{\top} x^{\init}\|_{\infty}}{\Xi} \) and \( \Xi = \max_{i} |u_{i} - \ell_{i}| \). Here, \( x_{i}^{\init} \) is defined as \( \frac{\ell_{i} + u_{i}}{2} \), and we set \( \wt{x}^{\init} \defeq \frac{1}{\beta}(b - \mA^{\top} x^{\init}) \). By adjusting the signs of the columns in \( \mA \), we can ensure that \( \wt{x}^{\init} \geq 0 \). If any component \( \wt{x}_{i}^{\init} = 0 \), we can eliminate that variable from consideration since it does not contribute to constructing the initial point. For the remaining components, we define \( \wt{\ell}_{i} = -\Xi \) and \( \wt{u}_{i} = 2\wt{x}_{i}^{\init} + \Xi \) (the terms \( -\Xi \) and \( +\Xi \) are included to guarantee \( \wt{u}_{i} > \wt{\ell}_{i} \)). We also set \( \wt{c} \defeq \frac{2\|c\|_{1}}{\d'} \), and we use the same symbol \( \wt{c} \) to denote a vector in \( \R^{n} \) where every entry is equal to this value.

Next, we consider the modified linear program:

\begin{align}\label{eq:modifylp}
\min_{\substack{\wt{\mA}^{\top} \begin{bmatrix} x \\ \wt{x} \end{bmatrix} = b \\ \ell \leq x \leq u \\ \wt{\ell} \leq \wt{x} \leq \wt{u}}} c^{\top} x + \wt{c}^{\top} \wt{x}. 
\end{align}

For this modified linear program \eqref{eq:modifylp}, the point

\[
\left(\begin{bmatrix} x^{\init} \\ \wt{x}^{\init} \end{bmatrix}, \begin{bmatrix} c \\ \wt{c} \end{bmatrix}, \mu\right)
\]

is \( \eps \)-centered with \( \mu = \frac{8m\|c\|_{1}\Xi}{\eps \delta'} \).
\end{definition}

Consequently, \cite{brand2021minimumcostflowsmdps} demonstrate that using a Laplacian solver a solution to the modified linear program \eqref{eq:modifylp} yields a near-feasible near-optimal solution for the original program \eqref{eq:generalLP}:

\begin{lemma}[Final Point of the Original LP, lemma 8.4 of \cite{brand2021minimumcostflowsmdps}]\label{lemma:finalpointmodifylp}
Assuming that the linear program in \eqref{eq:generalLP} has a feasible solution, and given an \( \eps \)-centered point for the modified LP in \eqref{eq:modifylp} with \( \mu = \frac{\d' \|c\|_{1} \Xi}{C n} \) for some sufficiently large constant \( C \) and any \( \d \leq 1 \), using a Laplacian solver for the Laplacian system $(\mA^\top {\Tau}^{-1}\Phi''(x)^{-1}\mA) x = b$, we can produce a point \(x^{\final} \) that satisfies the following conditions:

\[
c^{\top} x^{\final} \leq \min_{\substack{\mA^{\top} x = b \\ \ell \leq x \leq u}} c^{\top} x + \d \quad \text{and} \quad \|\mA^{\top} x^{\final} - b\|_{\infty} \leq \d,
\]

with \( \ell_{i} \leq x_{i}^{\final} \leq u_{i} \) for all \( i \).

\end{lemma}

In the following section, we will discuss how the modified linear program \eqref{eq:modifylp} and the \PathFollowing~algorithm are used to derive the results presented in \Cref{theorem:generalLP}.

\subsection{Proof of Correctness and Complexity} \label{subsection:proofIPM}

\begin{proof}[Proof of \Cref{theorem:generalLP}]
The algorithm presented thus far constructs a modified linear program of the type \eqref{eq:modifylp} and solves it using the \PathFollowing~algorithm (\Cref{algo:pathfollowing}, which utilizes \Cref{algo:twoplsstep} for its short-step iterations). This solution is subsequently applied to derive a near-optimal, near-feasible solution for the original linear program \eqref{eq:generalLP}.

As shown in \Cref{lemma:modifiedLP}, we can derive an \(\varepsilon\)-centered initial point \(\left(\begin{bmatrix}x^{\init}\\ \wt x^{\init} \end{bmatrix},\begin{bmatrix}c\\ \wt c \end{bmatrix},\mu^{\init}\right)\) for the modified linear program \eqref{eq:modifylp}, with \(\mu^{\init} = 8m \|c\|_{1} \Xi / \varepsilon \delta'\). To set up this modified LP in a two-party communication environment, one party, say Alice, can include the block matrix \(\begin{bmatrix}\beta \mI_n \end{bmatrix}\) in her part of \(\mA\). Calculating \(\beta\) requires Alice to determine \(\|b - \mA^{\top} x^{\init}\|_{\infty}\), which equals \(b - (\mA\pA)^\top (x\pA)^\init - (\mA\pB)^\top (x\pB)^\init\); Alice and Bob can compute this norm with \(O(nL \log(m))\) bits of communication. Finding \(\Xi\) only requires \(O(L)\) bits, as Alice and Bob simply communicate \(u_i\) and \(l_i\) for the largest \(|u_i - l_i|\). Therefore, calculating \(\beta\) requires a total of \(O(nL \log(m))\) bits. Setting \(\wt{c} = \frac{2 \|c\|_1}{\d'}\) requires \(\|c\|_1\), which can be communicated in \(O(L \log(m))\) bits since each component of \(c\) has bit complexity \(L\). Altogether, constructing the modified LP needs \(O(nL \log(m))\) bits of communication.

We then set \(\mu^{\final} = \frac{\delta' \|c\|_{1} \Xi}{C n}\) for a sufficiently large constant \(C\) and apply the \PathFollowing~algorithm, which yields \(\left(\begin{bmatrix} x^{\target}\\ \wt x^{\target} \end{bmatrix},\begin{bmatrix}s^{\target}\\ \wt s^{\target} \end{bmatrix},\mu^{\target}\right)\). By \Cref{lemma:pathfollowing}, this algorithm requires \(\tO(\sqrt{n} \log(\mu^\init/\mu^\final))\) iterations. As established in the previous section, setting \(\mu^\init = \frac{8m \|c\|_{1} \Xi}{\varepsilon \delta'}\) results in a total of \(\tO(\sqrt{n} L \log m)\) iterations.

Now, we analyze the bit complexity of each round of the \TwoPartyShortStep~algorithm. Since this algorithm operates on the modified LP, we need to account for the properties of the matrix \(\wt{\mA} = \begin{bmatrix} \mA \\ \beta \mI_{n} \end{bmatrix}\). 

First, observe that by adding a diagonal block matrix to \(\mA\), the number of non-zero entries per row remains unchanged. Thus, each row of \(\wt{\mA}\) still contains at most \(k\) non-zero entries.

Next, we examine the condition number \(\kappa(\wt{\mA})\) of the matrix \(\wt{\mA}\). Since \(\wt{\mA}^\top \wt{\mA} = \mA^\top \mA + \beta^2 \mI\), the condition number of \(\wt{\mA}\) is less than or equal to the condition number of \(\mA\):

\[
\kappa(\wt{\mA}) = \sqrt{\frac{\lambda_{\max} (\wt{\mA}^\top \wt{\mA})}{\lambda_{\min} (\wt{\mA}^\top \wt{\mA})}} = \sqrt{\frac{\lambda_{\max}(\mA^\top \mA) + \beta^2}{\lambda_{\min}(\mA^\top \mA) + \beta^2}} \leq \sqrt{\frac{\lambda_{\max}(\mA^\top \mA)}{\lambda_{\min}(\mA^\top \mA)}} = \kappa(\mA),
\]

where \(\lambda_{\max}\) and \(\lambda_{\min}\) denote the largest and smallest eigenvalues, respectively. Thus, adding the \(\beta \mI_n\) block effectively maintains or improves the condition number relative to \(\mA\).

In addition, since we need to calculate and maintain a spectral approximation of the matrix \(\bar{\mA} = {\Tau}^{-\frac{1}{2}}\Phi''(x)^{-\frac{1}{2}}\mA\), we must also bound its condition number.

The condition number \(\kappa(\mD^{\frac{1}{2}} \mA)\) for \(\mD = \Tau \Phi''(x)\) is given by:

\[
\kappa(\mD^{\frac{1}{2}} \mA) = \sqrt{\frac{\lambda_{\max} (\mA^\top \mD \mA)}{\lambda_{\min} (\mA^\top \mD \mA)}} \leq \sqrt{\frac{\max D_i}{\min D_i}} \kappa(\mA),
\]

where \(\mD_i\) represents the entries of the matrix \(\mD\). To ensure this bound, we examine the entries of \(\tau\) and \(\phi''(x)\):

1. The values \(\tau_i\) are regularized \(\ell_p\)-Lewis weights, defined as the solution to \(w(\mA) = \sigma(\mW^{\frac{1}{2} - \frac{1}{p}} \mA) + v\) (see \Cref{def:matrixlewis}). For calculating the leverage scores \(\sigma(\cdot)\), we use the procedure by \cite{ghadiri2024improvingbitcomplexitycommunication} (see \Cref{lemma:levScore}), which constructs leverage scores as powers of two within \([\frac{1}{2 m^2}, 1]\). The regularization vector \(v\) is set to \(\frac{n}{m} \mathbf{1}\).

2. The entries \(\phi''_i(x_i) = 1/(u_i - x_i)^2 +1/(x_i - l_i)^2 \geq 1/(u-l)^2\) provide a lower bound. Furthermore, as \cite{brand2021minimumcostflowsmdps} demonstrate, \(\log \Phi''(x)^{-1} \geq -\tO(L + \log \mu^\final + \log \|c\|_\infty)\).

Combining these properties gives us:

\[
\log \kappa(\bar{\mA}) \leq \tO(L \log m) \log \kappa(\mA).
\]

This ensures that the condition number of \(\bar{\mA}\) is efficiently bounded relative to the original matrix \(\mA\).

Now, we proceed with bounding the communication complexity of the \(\TwoPartyShortStep\) algorithm. In the following, assume \(\kappa = \kappa(\mA)\).

1. Setting the \(\ell_p\)-Lewis Weights: From \Cref{lemma:lewis-weight-communication}, the communication complexity for setting the \(\ell_p\)-Lewis weights in \Cref{line:twoipm:lewWeight} is \(\tO(nkL + n \log \kappa)\) bits.

2. Spectral Approximation and Maintenance: In \Cref{line:twoipm:H}, we compute and maintain a spectral approximation of the matrix \(\bar{\mA} = {\Tau}^{-\frac{1}{2}}\Phi''(x)^{-\frac{1}{2}}\mA\), which then facilitates the computation of \(\mH\). By \Cref{lemma:specApproxlemma}, constructing this spectral approximation requires \(\tO(nkL + nL \log \kappa \log m)\) bits. To maintain this spectral approximation throughout the algorithm, \Cref{lemma:inversemaintenance} implies an additional cost of \(\tO(nkL^2)\) bits across all iterations. 

3. Subprocedure Communications: In \Crefrange{line:twoipm:subp1}{line:twoipm:subp3}, each step requires the transmission of a vector of length \(n\), resulting in \(\tO(nL)\) bits of communication per step.

4. Other Steps: All remaining steps in \(\TwoPartyShortStep\) are either performed by the individual communication parties or the central coordinator, thus not contributing further to the communication cost.

Summing these complexities, the bit complexity of \(\TwoPartyShortStep\) over all iterations is:

\[
\tO(\sqrt{n} L \log m (nkL + nL \log \kappa \log m) + nkL^2).
\]

Therefore, as a component of \(\PathFollowing\), \(\TwoPartyShortStep\) requires a total communication complexity of:

\[
\tO(n^{1.5} L^2 (k + \log \kappa \log m) \log m) \quad \text{bits}.
\]

In the final step of \(\PathFollowing\), we approximately solve a Laplacian system. According to \Cref{lemma:specApproxEquivalency}, this can be achieved by computing a spectral approximation, similar to the spectral approximation step in \Cref{line:twoipm:H}. This approximation requires \(\tO(nkL + n \log \kappa)\) bits of communication.

Thus, the \(\PathFollowing\) algorithm provides \(\left(\begin{bmatrix} x^{\target}\\ \wt x^{\target} \end{bmatrix},\begin{bmatrix}s^{\target}\\ \wt s^{\target} \end{bmatrix}\right)\), a solution to the modified LP in \eqref{eq:modifylp}, with an overall communication complexity of \(\tO(n^{1.5} L^2 (k + \log \kappa \log m) \log m)\) bits. This solution is subsequently utilized by the procedure outlined in \Cref{lemma:finalpointmodifylp}, which applies a Laplacian solver to derive a solution \(x^\target\) for the original LP \eqref{eq:generalLP}. Thanks to \Cref{lemma:specApproxEquivalency}, this step is efficiently executed using a spectral approximation.

\end{proof}

\clearpage
\section{Communication Complexity of Minimum Cost Flow}\label{section:MCMFinTPCModel}

In this chapter, we explore the communication complexity of the minimum cost flow problem. The minimum cost flow problem is defined on a directed graph \( G = (V, E, u, c, d) \), where:

\begin{itemize}
    \item \( V \) is the set of vertices with \( |V| = n \),
    \item \( E \) is the set of directed edges with \( |E| = m \),
    \item \( u \in \R^m_{\geq 0} \) represents the capacities of the edges,
    \item \( c \in \R^m \) represents the cost per unit of flow on each edge,
    \item \( d \in \R^n \) represents the demands of each vertex. 
\end{itemize}

The objective is to find a flow \( f \in \R^m \) that satisfies the following conditions:

\begin{enumerate}
    \item \textbf{Flow Conservation and Demand Satisfaction}: For each vertex \( v \in V \), the net flow (the difference between the total incoming and outgoing flow) must match the vertex's demand \( d_v \). Mathematically:
    \[
    \sum_{e \in E^+(v)} f_e - \sum_{e \in E^-(v)} f_e = d_v,
    \]
    where \( E^+(v) \) and \( E^-(v) \) are the sets of edges entering and leaving vertex \( v \), respectively. Here, \( d_v \) denotes the demand at \( v \): \( d_v > 0 \) indicates that \( v \) requires excess incoming flow, while \( d_v < 0 \) indicates that \( v \) has excess outgoing flow.
    
    \item \textbf{Capacity Constraints}: The flow on each edge \( e \in E \) must satisfy:
    \[
    0 \leq f_e \leq u_e,
    \]
    ensuring that the flow does not exceed the edge’s capacity.
    
    \item \textbf{Objective}: Minimize the total cost of the flow, given by:
    \[
    \sum_{e \in E} c_e f_e = c^{\top} f.
    \]
\end{enumerate}

The solution to the minimum cost flow problem is a feasible flow \( f \) that respects edge capacities, satisfies all vertex demands, and minimizes the total cost.

\paragraph{Minimum Cost Flow in the Two-Party Communication Model}

In the two-party communication model, the problem is solved on the union graph of two communication parties. Each party knows some edges of the graph, along with their respective capacities and costs, while both parties share knowledge of the vertex set \( V \) and their demands \( d \). The goal is to compute a feasible optimal flow \( f \in \R^m \) with minimal communication between the two parties.

The main result of this chapter states the communication complexity of this problem:

\theoremmincostflow*

First, note that this problem can be precisely formulated as a linear program with two-sided constraints:

\begin{align}
    \min_{\substack{\mA^\top f = d \\ 0 \le f_e \le u_e \, \forall e \in E}} c^\top f, \label{eq:mincostFlow}
\end{align}

where \( \mA \in \{-1, 0, 1\}^{m \times n} \) is the incidence matrix of \( G \). Specifically, for each edge \( e = (u, v) \in E \), the entries of \( \mA \) are defined as \( \mA_{e,u} = -1 \) and \( \mA_{e,v} = 1 \). 

Note that each edge known by Alice or Bob corresponds to a row of the matrix \( \mA \). This formulation indicates that we aim to solve \eqref{eq:mincostFlow} in a setting analogous to the general linear program setting of \eqref{eq:generalLP}. To utilize the IPM introduced in the last chapter, we need to bound two key quantities:
\begin{itemize}
    \item \( k \), the maximum number of non-zero entries in \( \mA \),
    \item \( \kappa \), the condition number of \( \mA \).
\end{itemize}

\subsection{Analysis Tools} \label{subsection:analysistools}

First, it is easy to see that \( k = 2 \) for the incidence matrix, since each row has exactly two non-zero entries: one \( 1 \) and one \( -1 \). 

Next, we show that \( \kappa \leq O(\sqrt{n}) \), where \( \kappa \) is the condition number of the incidence matrix \( \mA \).

\begin{lemma}\label{lemma:conditionnumMinCost}
    Let \( \mA \in \{-1, 0, 1\}^{m \times n} \) be the incidence matrix of the \( n \)-vertex, \( m \)-edge, directed graph \( G \). Then \( \kappa(\mA) \leq O(n\sqrt{n}) \).
\end{lemma}

\begin{proof}[Proof of \cref{lemma:conditionnumMinCost}]
    First, note that by definition, we have:
    \[
        \kappa(\mA) = \frac{\sigma_{\max}(\mA)}{\sigma_{\min}(\mA)},
    \]
    where \( \sigma_{\max} \) and \( \sigma_{\min} \) denote the largest and smallest singular values of \( \mA \), respectively. Furthermore, let \( \lambda_{\max} \) and \( \lambda_{\min} \) be the largest and smallest non-zero eigenvalues of the matrix \( \mL \coloneqq \mA^\top \mA \). It follows that:
    \[
        \sigma_{\max}(\mA) = \sqrt{\lambda_{\max}} \quad \text{and} \quad \sigma_{\min}(\mA) = \sqrt{\lambda_{\min}}.
    \]

    Since \( \mA \) is the incidence matrix of \( G \), the matrix \( \mL \) is the Laplacian matrix of the undirected graph \( G' \), obtained by ignoring edge directions and weights in \( G \). Therefore, we need to bound the largest and smallest non-zero eigenvalues of the Laplacian matrix of \( G' \).

    The smallest non-zero eigenvalue of \( \mL \), \( \lambda_{\min} \), is the algebraic connectivity of \( G' \). It is known that for a simple connected graph with \( n' \) vertices and diameter \( D \):
    \[
        \lambda_{\min} \geq \frac{1}{n'D},
    \]
    as stated in \cite{lowerBoundKappa}, Theorem 4.2. Since this holds for each connected component of \( G' \), we conclude:
    \[
        \lambda_{\min} \geq \frac{1}{n^2}.
    \]

    On the other hand, it is known (see \cite{upperBoundKappa}, Theorem 2) that:
    \[
        \lambda_{\max} \leq \max\{d(u) + d(v) \mid (u, v) \in E(G')\} \leq 2n,
    \]
    where \( d(u) \) is the degree of vertex \( u \) in \( G' \).

    Thus, the condition number satisfies:
    \[
        \kappa(\mA) = \sqrt{\frac{\lambda_{\max}}{\lambda_{\min}}} \leq O(n\sqrt{n}).
    \]
\end{proof}

Lastly, since our goal is to obtain an exactly optimal and feasible flow, rather than a near-optimal, near-feasible one, we introduce the following lemma from \cite{brand2021bipartitematchingnearlylineartime}:

\begin{lemma}[Lemma 8.10 of \cite{brand2021bipartitematchingnearlylineartime}]\label{lemma:isolation}
    Let \( \Pi = (G, d, c) \) be an instance of the minimum-cost flow problem, where \( G \) is a directed graph with \( m \) edges, the demand vector \( d \in \{-W, \dots, W\}^V \), the cost vector \( c \in \{-W, \dots, W\}^E \), and the capacity vector \( u \in \{0, \dots, W\}^E \).
    
    Let the perturbed instance \( \Pi' = (G, d, c') \) be such that \( c'_e = c_e + z_e \), where \( z_e \) is a random number from the set 
    \[
    \left\{ \frac{1}{4m^2 W^2}, \dots, \frac{2mW}{4m^2 W^2} \right\}.
    \]
    Let \( x' \) be a feasible flow for \( \Pi' \) whose cost is at most \( \mathrm{OPT}(\Pi') + \frac{1}{12m^2 W^3} \), where \( \mathrm{OPT}(\Pi') \) is the optimal cost for problem \( \Pi' \). Then, with probability at least \( 1/2 \), there exists an optimal feasible and integral flow \( x \) for \( \Pi \) such that \( \| x - x' \|_\infty \leq \frac{1}{3} \).
\end{lemma}

We are now ready to prove \Cref{theorem:mincostflow}. The following proof is inspired by Theorem 1.4 of \cite{brand2021minimumcostflowsmdps}. While we modify certain steps, the proof remains largely similar to that of Theorem 1.4 in \cite{brand2021minimumcostflowsmdps}.

\subsection{Proof of Correctness and Complexity} \label{subsection:mincostproof}

\begin{proof}[Proof of \cref{theorem:mincostflow}]
    Given the input graph \( G = (V, E, u, c) \), we first perturb the edge costs \( c \) to \( c' \) according to \cref{lemma:isolation}, obtaining the perturbed graph \( G' = (V, E, u, c') \). This step does not require any communication, as we assume both parties share randomness.

    In the following, and consistent with \cref{lemma:isolation}, we assume that \( W \) bounds the entries of \( u \), \( c \), and \( d \).

    Next, we consider the linear program \eqref{eq:mincostFlow}, which describes this problem. Similar to the algorithm described in the proof of \cref{theorem:generalLP} (see \cref{subsection:proofIPM}), we set 
    \[
    \mu^\init = 100 m^2 W^3 \varepsilon^{-1}, \quad \mu^\target = \frac{1}{\mathrm{poly}(mW)},
    \]
    modify the LP to create a new LP for a graph \( \wt{G}' \) with a trivial initial solution, and use the \PathFollowing~algorithm to compute a near-optimal, near-feasible solution 
    \[
    \begin{bmatrix} f'^\apxfinal \\ \wt{f}'^\apxfinal \end{bmatrix}
    \]
    for \eqref{eq:mincostFlow}. As proven by \cite{brand2021minimumcostflowsmdps} (see Lemma 7.7, \cite{brand2021minimumcostflowsmdps}), this solution's entries differ by at most \( 1/(mW)^{10} \) from an exact feasible flow 
    \[
    \begin{bmatrix} f'^\final \\ \wt{f}'^\final \end{bmatrix}
    \]
    for \( \wt{G}' \), which we would obtain if an exact Laplacian solver were used in the final step of \PathFollowing. They also prove that 
    \[
    c^\top f'^\final + \wt{c}^\top \wt{f}'^\final \leq \opt(\wt{G}') + \frac{1}{12 m^2 W^3},
    \]
    where \( \opt(G) \) denotes the optimal value of LP \eqref{eq:mincostFlow} for graph \( G \). Note that \( \wt{G}' \) is effectively a graph constructed from \( G' \) by adding a bi-directional star rooted at a new vertex to \( G' \).

    The results of \cite{brand2021minimumcostflowsmdps} establish that if LP \eqref{eq:mincostFlow} has a feasible solution, then the auxiliary flow component satisfies \( \|\wt{f}'^\final\|_\infty < 0.1 \). Consequently, if \( \|\wt{f}'^\final\|_\infty \geq 0.1 \), we can confidently conclude that the chosen demand vector \( d \) renders the problem infeasible for \( G \). 

    If \( \|\wt{f}'^\final\|_\infty < 0.1 \), we proceed by rounding the entries of the approximate solution 
    \[
    \begin{bmatrix} f'^\apxfinal \\ \wt{f}'^\apxfinal \end{bmatrix}
    \]
    to the nearest integers to obtain 
    \[
    \begin{bmatrix} f^\final \\ \wt{f}^\final \end{bmatrix}.
    \]
    At this stage, \cref{lemma:isolation} guarantees that, with probability at least \( 1/2 \), there exists an integral optimal solution 
    \[
    \begin{bmatrix} f^\opt \\ \wt{f}^\opt \end{bmatrix}
    \]
    to the problem for the perturbed graph \( \wt{G}' \), and this solution differs entry-wise from 
    \[
    \begin{bmatrix} f'^\final \\ \wt{f}'^\final \end{bmatrix}
    \]
    by no more than \( 1/3 \). Since \( \begin{bmatrix} f^\opt \\ \wt{f}^\opt \end{bmatrix} \) is guaranteed to differ from 
    \[
    \begin{bmatrix} f'^\apxfinal \\ \wt{f}'^\apxfinal \end{bmatrix}
    \]
    entry-wise by less than \( 1/2 \), 
    \[
    \begin{bmatrix} f^\final \\ \wt{f}^\final \end{bmatrix} = \begin{bmatrix} f^\opt \\ \wt{f}^\opt \end{bmatrix}
    \]
    with probability at least \( 1/2 \). 

    Moreover, because \( \|\wt{f}'^\final\|_\infty < 0.1 \), it follows that \( \wt{f}^\final = 0 \). This implies that \( f^\final \) is a feasible flow for the original graph \( G \), with a total cost given by 
    \[
    c^\top f^\final = c^\top f^\final + \wt{c}^\top \wt{f}^\final = \opt(\wt{G}'),
    \]
    where \( \opt(\wt{G}') \leq \opt(G) \). The inequality holds because adding edges to a graph cannot increase the optimal value. Therefore, \( f^\final \) is indeed the optimal feasible flow for \( G \).

    Finally, to ensure the algorithm succeeds with high probability, we can repeat it \( O(\log n) \) times, thereby amplifying the success probability. 

    Regarding communication complexity, the analysis follows that of the \cref{theorem:generalLP}. Note that the entries of \( u \), \( c \), and \( d \) are bounded by \( W = 10n \|u\|_\infty \|c\|_\infty \), so \( L \in \tO(\log(\|u\|_\infty \|c\|_\infty)) \). Additionally \( m \in O(n^2) \).

    Constructing the modified LP requires \( O(nL \log m) = \tO(n \log(\|u\|_\infty \|c\|_\infty)) \) bits of communication. Executing the \PathFollowing~algorithm involves \( \tO(\sqrt{n} \log(\mu^\init / \mu^\final)) = \tO(\sqrt{n} \log(\poly mW)) = \tO(\sqrt{n} \log(\|u\|_\infty \|c\|_\infty)) \) iterations, with each iteration requiring \( \tO(nkL + nL \log \kappa \log m) = \tO(n \log(\|u\|_\infty \|c\|_\infty)) \) bits of communication. Consequently, the total bit complexity for \PathFollowing~is \( \tO(n^{1.5} \log^2(\|u\|_\infty \|c\|_\infty)) \) bits.

    Additionally, solving a Laplacian system, as per \cref{lemma:specApproxlemma}, demands \( \tO(nkL + n \log \kappa) = \tO(n \log(\|u\|_\infty \|c\|_\infty)) \) bits of communication.

    In summary, the overall bit complexity for the minimum-cost flow problem amounts to \( \tO(n^{1.5} \log^2(\|u\|_\infty \|c\|_\infty)) \) bits of communication.
      
\end{proof}

\subsection{Communication Complexity of Maximum Flow} \label{subsection:maxflow}

In this section, we discuss the communication complexity of the maximum flow problem. The maximum flow problem is defined on a directed graph \( G = (V, E, u) \), where:
\begin{itemize}
    \item \( V \) is the set of vertices with \( |V| = n \),
    \item \( E \) is the set of directed edges with \( |E| = m \),
    \item \( u \in \R^m_{\geq 0} \) represents the capacities of the edges.
\end{itemize}

The objective is, for given vertices \( s \) (source) and \( t \) (sink), to find a flow \( f \in \R^m \) that satisfies the following conditions:

\begin{enumerate}
    \item \textbf{Flow Conservation}: For each vertex \( v \in V \setminus \{s, t\} \), the net flow (difference between incoming and outgoing flows) must be zero. Mathematically:
    \[
    \sum_{e \in E^+(v)} f_e - \sum_{e \in E^-(v)} f_e = 0,
    \]
    where \( E^+(v) \) and \( E^-(v) \) are the sets of edges entering and leaving vertex \( v \), respectively.

    \item \textbf{Capacity Constraints}: For each edge \( e \in E \), the flow must satisfy:
    \[
    0 \leq f_e \leq u_e,
    \]
    ensuring that the flow does not exceed the edge's capacity.

    \item \textbf{Maximization of Total Flow}: The total flow \( F \) from \( s \) to \( t \), defined as:
    \[
    F = \sum_{e \in E^-(s)} f_e - \sum_{e \in E^+(s)} f_e = \sum_{e \in E^+(t)} f_e - \sum_{e \in E^-(t)} f_e,
    \]
    is maximized.
\end{enumerate}

The solution to the maximum flow problem is a flow \( f \) that respects edge capacities, satisfies flow conservation, and maximizes \( F \), the total flow from \( s \) to \( t \).

Similar to minimum cost flow problem, in the two-party communication model, maximum flow problem is solved on the union graph of two communication parties. Each party knows some edges of the graph, along with their respective capacities. Both parties know the vertex set $V$. The goal is to find the maximum feasible flow $f \in \R^m$ with minimal communication between the two parties. 

The main result of this section is the following:

\theoremmaxflow*

First, note that the maximum flow problem can be reduced to the minimum cost flow problem. Modify \( G \) to \( G' \) by adding an \( s \)-\( t \) edge \( e' = (s, t) \) with \( u_{e'} = F \) for a sufficiently large \( F \), e.g., \( F = \sum_{e \in E(G)} u_e \), and \( c_{e'} = 1 \). The costs of the remaining edges are set to \( 0 \). Additionally, set \( d_s = F \), \( d_t = -F \), and \( d_v = 0 \) for the remaining vertices. By computing the minimum-cost flow in \( G' \), we compute the maximum flow in \( G \), as the minimum cost flow in \( G' \) would avoid sending flow through \( e' \) as much as possible. Specifically, for an optimal minimum cost flow \( f' \in \mathbb{Z}^{m+1} \) of \( G' \), the flow \( f \in \mathbb{Z}^m \) with \( f_e = f'_e \) for all \( e \in E(G) \) is a maximum flow of \( G \) with flow value \( F - f'_{e'} \).

Considering \cref{theorem:mincostflow}, we obtain an algorithm for the maximum flow problem with a bit complexity of \( \tO(n^{1.5} \log^2 \|u\|_\infty) \). However, as \cite{brand2021minimumcostflowsmdps} demonstrate, using a standard scaling technique (Section 6 of \cite{Ahuja1991DistanceDirectedAP}, Chapter 2.6 of \cite{Williamson_2019}), it is possible to reduce this bit complexity to \( \tO(n^{1.5} \log \|u\|_\infty) \). The main idea is to use the algorithm for minimum-cost flow in \( \log \|u\|_\infty \) iterations, where in each iteration the edge capacities are small.

\begin{proof}[Proof of \cref{theorem:maxflow}]

    First, we introduce the algorithm presented by \cite{brand2021minimumcostflowsmdps}.
    
    For any graph \( G \) and a flow \( f \) within \( G \), define \( G_f \) as the residual graph of \( G \) relative to the flow \( f \). Additionally, let \( G_f(\Delta) \) represent the graph derived from \( G_f \) by excluding all edges whose capacities in \( G_f \) are less than \( \Delta \). With these definitions in place, consider the algorithm below:

    \begin{algorithm2e}[h]
    \caption{Algorithm for computing the maximum flow $f \in \Z^m$ of a directed graph $G = (V, E)$ with edge capacities $u \in \R^{m} $ \label{algo:maxflow}}
    \SetKwInOut{Input}{Input}
    \SetKwInOut{Output}{Output}
    \Input{Graph $G = (V,E)$ with edge capacities $u \in \R^{m}$}
    \Output{A flow $f \in \R^m$ of graph $G$ with maximum flow value}
    
    \SetKwProg{Procedure}{Procedure}{}{}
        \Procedure{\textsc{MaxFlow}$(G, u)$}{
            Set $f = 0$ and $\Delta = 2^{\left\lfloor \log_{2}\|u\|_{\infty}\right\rfloor }$. \\

            \While {$\Delta \geq 1$} {
                Set $G' = G_f(\Delta)$. \\
            
                In $G'$, for each $e \in E(G')$, set the edge capacity $u'_e = \left\lfloor \min(u_e, 2m\Delta) / \Delta \right\rfloor$. \label{line:maxflow:edgecap} \\

                As explained above, find the maximum flow $f'$ of $G'$ (extended to $m$-dimensions). \\

                Set $f = f + \Delta \cdot f'$ and $\Delta = \Delta / 2$. \\
            }
            \Return~ $f$.
        }
    \end{algorithm2e}

    This algorithm correctly computes the maximum flow because, after the final iteration when \( \Delta = 1 \), no augmenting path remains in \( G_f \). Furthermore, at the start of each iteration in the while loop of \Cref{algo:maxflow}, the maximum flow value in \( G_f(\Delta) \) is at most \( 2m\Delta \). 

    This is obvious in the first iteration. For subsequent iterations, this holds because the previous iteration ensures that the maximum flow value in \( G_f(2\Delta) \) is zero. Since \( G_f \) can be constructed from \( G_f(2\Delta) \) by adding at most \( m \) edges, each with a capacity less than \( 2\Delta \), the maximum flow value in \( G_f \) is at most \( 2m\Delta \). 

    As a result, in each iteration in \Cref{line:maxflow:edgecap}, the edge capacities can be safely capped at \( 2m\Delta \) without affecting the maximum flow. These capacities are then scaled down to positive integers less than \( 2m \). Since the subsequent steps do not involve communication, the bit complexity of each iteration is \( \tO(n^{1.5} \log m) = \tO(n^{1.5}) \). Given that there are \( \tO(\log \|u\|_\infty) \) iterations, the total bit complexity amounts to \( \tO(n^{1.5} \log \|u\|_\infty) \) bits.
    
\end{proof}

\clearpage

\section{Acknowledgments}
This paper is based on the bachelor’s thesis ``Exploring the Maximum Flow Problem in Non-Sequential Settings'' by Hossein Gholizadeh, written in December 2024 as an external thesis at the Max Planck Institute for Informatics, Saarbrücken, and the Karlsruhe Institute of Technology. We thank %
Jan van den Brand for suggesting the main idea of this work, and Marvin Künnemann and Danupon Nanongkai for their feedback and evaluation.

\begin{refcontext}[sorting=nyt]
\printbibliography[heading=bibintoc]
\end{refcontext}

\end{document}